\documentclass[reqno]{amsart} 
\usepackage{graphicx, amsthm}%
\usepackage[foot]{amsaddr}
\usepackage{dcolumn}
\usepackage{framed}
\usepackage{bm}
\usepackage{color}
\usepackage{extpfeil}
\usepackage{titlesec}
\usepackage{ulem}

\usepackage{enumitem}
\usepackage{hyperref}
\titleformat{\section}{\Large\bfseries}{}{0.3em}{}
\titleformat{\subsection}{\normalsize\bfseries}{}{0em}{}
\newtheorem{theorem}{\textbf{Theorem}}
\newtheorem{proposition}{\textbf{Proposition}}
\newtheorem{lemma}{\textbf{Lemma}}
\newtheorem{corollary}{\textbf{Corollary}} 
\renewcommand{\eqref}[1]{Eq. (\ref{#1})}

\title{Exact first passage time distribution for second-order reactions in chemical networks}

\usepackage{fancyhdr}
\author{Changqian Rao$^{1,2}$
}

\author{David Waxman$^{3,4,5,6}$ 
}

 \author{Wei Lin$^{1,2,3}$
}
 \author{Zhuoyi Song\textsuperscript{*} $^{3,4,5,6}$
}
 \email{songzhuoyi@fudan.edu.cn}

  \address{ $^1$ School of Mathematical Sciences, Fudan University, Shanghai, 200433, China}
  \address{$^2$ 
Research Institute of Intelligent Complex Systems, Fudan University, Shanghai, 200433, China
}

  \address{ $^3$
Institute of Science and Technology for Brain-Inspired Intelligence, Fudan University, 
Shanghai, 200433, China
}
  \address{ $^4$
Key Laboratory of Computational Neuroscience and Brain-Inspired Intelligence (Fudan University), Ministry of Education, China
}
  \address{ $^5$
MOE Frontiers Center for Brain Science, Fudan University, Shanghai 200433, China
}
  \address{ $^6$
Zhangjiang Fudan International Innovation Center, Shanghai, China
}
\begin{document}

\date{\today}

\begin{abstract}

The first passage time (FPT) is a generic measure that quantifies when a random quantity reaches a specific state. We consider the FTP distribution in nonlinear stochastic biochemical networks, where obtaining exact solutions of the distribution is a challenging problem. Even simple two-particle collisions cause strong nonlinearities that hinder the theoretical determination of the full FPT distribution. Previous research has either focused on analyzing the mean FPT, which provides limited information about a system, or has considered time-consuming stochastic simulations that do not clearly expose causal relationships between parameters and the system's dynamics. This paper presents the first exact theoretical solution of the full FPT distribution in a broad class of chemical reaction networks involving $A + B \rightarrow C$ type of second-order reactions. Our exact theoretical method outperforms stochastic simulations, in terms of computational efficiency, and deviates from approximate analytical solutions. Given the prevalence of bimolecular reactions in biochemical systems, our approach has the potential to enhance the understanding of real-world biochemical processes.

\end{abstract}


\maketitle


\section{Introduction}

The first passage time (FPT) is a fundamental concept that is used to analyze the behavior and dynamics of stochastic processes\cite{benichou2011intermittent,metzler2014first,redner2001guide}. In biochemical reaction networks, the FPT is a key quantity that refers to the time when a specific event or state first occurs within the network\cite{bressloff2014stochastic,iyer2016first}.
Examples of the specific event include reaction completion\cite{bel2009simplicity}, binding or unbinding events\cite{newby2016first}, protein translocation\cite{RevModPhys.85.135}, and state transitions\cite{dai2015first}. Analyzing the FPT for these events is not just a theoretical exercise; the FPT can provide detailed insight into the timing, efficiency, and reliability of the underlying biochemical reaction system\cite{ham2024stochastic,polizzi2016mean}. This insight is crucial because it enables not only the understanding of regulatory mechanisms\cite{huang2021relating}, but also their manipulation\cite{ghusinga2017first}, thereby opening up new possibilities in biochemistry and chemical kinetics.  

When counts of the constituent molecules are low, stochasticity and discreteness are inescapable features of chemical kinetics\cite{raj2008nature}. In this context, the stochastic properties of the FPT require a characterization in terms of its probability distribution. However, past theoretical work has primarily focused on deriving the mean FPT and the global FPT\cite{benichou2010geometry,condamin2007first}. This is because the FPT distribution is hard to measure experimentally; the stochastic timings are disguised in cell population measurements due to cell-cell variabilities, and precisions in single-cell measurements are limited by experimental technologies\cite{iyer2014scaling,wang2010robust}. Attention has shifted recently to focus on obtaining the full FPT distribution beyond the mean\cite{godec2016first, woods2024analysis}. This distribution provides much more information about the underlying biochemical system. For example, Thorneywork et al. demonstrate that a purely dynamic measurement of the full FPT distribution uncovers that a short-time, power-law regime of the distribution, rather than the mean FPT, reflects the number of intermediate states in an underlying potential energy landscape\cite{thorneywork2020direct}. Therefore, mathematical methodologies for estimating and analyzing the FPT distribution are highly desirable, and complement advances in measurement technologies of the FPT distribution in biochemical systems.

Traditionally, estimating the FPT distributions depends on solving the underlying chemical master equations (CMEs)\cite{iyer2016first}, which are the primary modelling approach of stochastic biochemical systems\cite{gillespie2013perspective,iyer2016first}. Typically, CME solutions can be simulated\cite{cao2006efficient}, solved approximately\cite{smadbeck2013closure,cao_linear_2018}, or solved exactly\cite{ishida1964stochastic, mcquarrie1964kinetics}. Simulation approaches, such as the Gillespie algorithm\cite{gillespie2013perspective}, approximate the solutions of the CME by generating many realizations of the associated Markov process. Such simulations trade time efficiency for accuracy, since producing tens of thousands of sample paths, to expose the underlying distribution, takes time\cite{rao2023analytical,jia2024holimap}. 

Approximate methods were developed to solve CMEs, trading off estimation accuracy for time efficiency\cite{smadbeck2013closure,cao_linear_2018}. There are two main classes of such methods: (i) closure schemes\cite{smadbeck2013closure} and (ii) linear mapping methods\cite{cao_linear_2018}. Under closure schemes, the solutions to CMEs are obtained by approximating higher-order moments of the solution by nonlinear functions of lower-order moments, thereby leading to tractable equations. Under linear mapping, bimolecular reactions are approximated by first-order reactions, allowing simpler, solvable systems. However, even with such approximations, it is a formidable task to elucidate the causal mechanisms regulating FPT distributions, as exhaustive parameter searches are usually unrealistic, and the complex compensatory effects of parameter variations are difficult to clarify\cite{rao2023analytical}. 

Exact theoretical expressions for the FPT distributions are highly desired, as they may be essential to identifying and quantifying the causal regulatory mechanisms\cite{venugopal_modelling_nodate,moss_identification_2008,vizan_controlling_2013}. However, this depends on determining time-dependent solutions for the CMEs, which are only known for specific cases\cite{peccoud1995markovian,ramos2011exact}.To the best of our knowledge, general, exact time-dependent solutions of CMEs are known only for simple reaction systems - with zero and first-order reactions, i.e., \textit{linear systems}\cite{jahnke2007solving,smadbeck2013closure}.

Time-dependent solutions of CMEs are still challenging for general, non-linear, biochemical reaction networks\cite{anderson_time-dependent_2020}. Even widely presented bimolecular reactions, which constitute one of the simplest core building blocks of a biochemical reaction network, involve highly nonlinear models, making exact CME solutions non-tractable\cite{gillespie2013perspective}. As a result, for biochemical networks with second-order reactions, time-varying CME solutions are only available for specific cases, such as highly simplified systems with only a few states\cite{ishida1964stochastic, mcquarrie1964kinetics, renyi1954betrachtung, schnoerr2017approximation}, or reversible bimolecular reaction in which the transition matrix has a tridiagonal format\cite{laurenzi2000analytical,smith2015general}.  

General time-dependent CME solutions are still unknown for systems with bimolecular reactions\cite{gillespie2013perspective}. Matters are even more challenging if there is a large state space and if non-constant reaction rates are involved\cite{rao2023analytical,gonze_goodwin_2013}. Generally, second-order reactions can be grouped into two types: $A + A \rightarrow C$ and $A + B \rightarrow C$, and mathematical analysis treats these two types of second-order reactions separately\cite{ishida1964stochastic, mcquarrie1964kinetics, renyi1954betrachtung}. In earlier work, we derived an exact FPT distribution for a general class of biochemical networks involving an $A + A \rightarrow C$ second-order reaction downstream of two zero/first-order reactions\cite{rao2023analytical}. In the present work, we report the first exact distribution of the FPT for a general class of chemical reaction networks with an $A + B \rightarrow C$ second-order reaction that is downstream of various zero/first-order reactions.

We note that Anderson et al. have recently provided the first exact time-dependent distribution for a general class of reaction networks with higher-order complexes\cite{anderson_time-dependent_2020}. They find that the time-varying solutions of the CMEs will maintain a Poisson-product form. However, this result requires that: (i) the chemical reaction system initially has such a Poisson-product form, and (ii) the system has to satisfy a dynamical and restricted (DR) condition. This DR condition requires that, for any higher-order reactant pairs, the production and consumption are \textit{balanced}. Generally, the DR condition is very restrictive, and only applies to specific reaction formats with specific kinetic rates and initial conditions. In this work, we present results that are not subject to the DR condition - by allowing the mean of molecular numbers to follow a stochastic process, as described by a stochastic differential equation (SDE) 
rather than an ordinary differential equation (ODE).

Our exact analytical results are not only novel and efficient, but also highly practical. They are much more time efficient than stochastic simulations, and are much more accurate than traditional approximation methods, such as linear mapping methods (LMA)\cite{cao_linear_2018}. As we demonstrate, our methods have applications in diverse chemical reaction networks from different areas. These include gene regulatory networks\cite{jia2024holimap} and multi-step transition models\cite{huang2021relating}. Our results thus have broad applicability and real-world significance. 

\section{Problem formulation}\label{sec:pro_formu}

We consider a biochemical system with $N$ chemical species $\bm{S} = [S_1,S_2,...,S_N]^{\top}$ whose molecules can undergo $M+1$ chemical reactions, $R_m$, with $m=0,1,2,3,....,M$. There are $M$ zero or first-order reactions followed by one second-order reaction, which is of the type  $A + B \rightarrow C$. 

We label the second-order reaction by $m=0$ and the other reactions (zero- or first-order) by $m=1,2,...M$. The molecules $A$ and $B$ are the reactants of the second-order reaction, which can be any two different molecular species. We denote these two reactants as $S_1$ and $S_2$ without loss of generality. 

We shall use the following general notation for such a biochemical system: 

\begin{equation}
\begin{split}\label{Cas1}
\bm{y}_m\cdot \bm{S}\overset{a_{m}(t)}{\rightarrow}& \bm{y'}_m\cdot \bm{S},\\
S_1+S_2\overset{a_{0}(t)}{\rightarrow}&*
\end{split}
\end{equation}
where $m= 1,\dots,M$ label the $m$'th zero or first-order reaction, while $\bm{y}_m$ and $\bm{y}_m'$ are the stoichiometric coefficient vectors of the $m$'th reaction. 
The reaction rate 'constants' are written $a_{m}(t)$ and $a_{0}(t)$, and we have assumed these vary with the time, $t$. The resulting analysis thus includes a broad class of non-linear time-varying biochemical systems with second-order reactions.

Because, in the above biochemical system, the first $m$ reactions are zero or first-order, we require that the sum of the absolute values of all components of $\bm{y}_m$, and separately of $\bm{y}_m'$, should not exceed one; we write these conditions as $\|\bm{y}_m\|_1\leq 1$ and $\| \bm{y}'_m\|_1\leq 1$, respectively.

To begin a derivation of the exact FPT distribution of the second-order reaction, we first define it mathematically. We adopt a standard approach by adding an additional species $S_0$, the number of which counts the time of occurrence of the second-order reaction\cite{iyer2016first}. We denote the number of $S_0$ molecules by $x_0 \equiv x_0(t)$, then $x_0$ is a non-decreasing function of $t$. Thus, we modify the system in \eqref{Cas1} to: 
\begin{equation}
\begin{split}\label{apb_reac}
\bm{y}_m\cdot \bm{S}\overset{a_{m}(t)}{\rightarrow}& \bm{y'}_m\cdot \bm{S},\\
S_1+S_2\overset{a_{0}(t)}{\rightarrow}&S_0.
\end{split}
\end{equation}
If at any time $t$ the second order reaction has not occurred, we have $x_0(t)=0$, and
equivalently the FPT exceeds $t$. Therefore, the complementary cumulative probability distribution of the FPT can be represented as: 

\begin{equation}
\begin{split}\label{fpt_aux}
{\rm P}(FPT>t) = P(x_0(t) = 0). \\
\end{split}
\end{equation}

\section{Result}
\subsection{Analytical representation for the auxiliary chemical master equation}

For the system given in \eqref{apb_reac}, we adopt the following notation:

\begin{enumerate}
    \item $\mathbf{X}=\left[\begin{array}{c}
         x_{0}  \\
         \mathbf{x} 
    \end{array}\right]$ denotes the complete state vector of the system, representing the numbers of different species of $\bm{S}^*=\left[\begin{array}{c}
         S_0  \\
         \bm{S} 
    \end{array}\right]$. The quantity $\mathbf{x}=[x_{1},\cdots,x_{N}]^{\top}$ denotes the numbers of different species of $\bm{S}$.

\item $\bm{Y}_m=\left[
         y_{m,0},
         \bm{y}_m 
    \right]$ and $\bm{Y'}_m=\left[
         y'_{m,0},  
         \bm{y'}_m 
    \right]$ are stoichiometric coefficient vectors of the $m$'th reaction in system \eqref{apb_reac}, and $\bm{y}_m=[y_{m,1},y_{m,2},\dots,y_{m,N}]$, $\bm{y'}_m=[y'_{m,1},y'_{m,2},\dots,y'_{m,N}]$ denote parts of stoichiometric coefficient vectors, excluding $y_{m,0}$ and $y'_{m,0}$ for the auxiliary species $S_0$.

\item $a_{0}(t)$ denotes the rate `constant' at time $t$ of the second-order reaction,
while $a_{m}(t)$ ($m=1,\dots,M$) denotes the rate `constant' at time $t$ of the $m$'th reaction
 (which is either a zero or first-order reaction).

\end{enumerate}

The CMEs for the system given in \eqref{apb_reac} correspond to a differential equation for the probability distribution of $\bm{X}$ at time $t$\cite{gillespie2013perspective}: 

\begin{small}
\begin{equation}
\label{general_cme_element}
\displaystyle\frac{{\rm d} P(\bm{X}, t )}{{\rm d} t} \! =\!\sum\limits_{k=0}^{M}\left[ P\left(\bm{X}-\bm{v}_k, t \right) c_k\left(\bm{X}-\bm{v}_k,t\right)\!-\! P(\bm{X}, t ) c_k(\bm{X},t)\right],
\end{equation}
\end{small}
where $\bm{v}_k$ is the transition vector for the $k$'th reaction, and $c_k(\bm{X},t)$ is the propensity function for the $k$'th reaction, i.e., the probability that the $k$'th reaction occurs in state $\bm{X}$, at time $t$. The 
quantity $c_k(\bm{X},t)$ equals $a_k(t)\bm{X}^{\bm{Y}_k}$, where, for any vectors with $d$ components, a vector to the power of a
vector is defined by $\mathbf{u}^{\mathbf{v}}\overset{\rm def}{=}\prod^d_{i=1}u_i^{v_i}$ and we adopt the convention that $0^0=1$. Therefore, for $k=0$, the reaction is given in \eqref{apb_reac}, thus, $c_0(\bm{X},t)=a_0(t)x_1x_2$. For $k=1,\dots, M$, the reactions are first or lower-order, and thus $c_k(\bm{X},t)$, $k=1,\dots, M$ are linear functions of states.

We next present a theorem that gives the exact solution of \eqref{general_cme_element} 
in terms of the following variables/notation:

\begin{enumerate}[label=(\roman*)]

\item $\boldsymbol{\lambda}(t)=[\lambda_{1}(t),\dots
,\lambda_{N}(t)]^\top$ is a column vector containing the mean of all numbers of the species in $\bm{S}$ at
any time $t$ ($t\geq0$), and $\boldsymbol{\Lambda}(t)= \left[\begin{array}{c}
     \lambda_{0}(t)  \\
     \boldsymbol{\lambda}(t) 
\end{array}  \right]$ is a column vector containing the mean of all numbers of the species in $\bm{S^*}$ at
any time $t$ ($t\geq0$).

\item 
$\boldsymbol{\Lambda}(0)= \left[\begin{array}{c}
     \lambda_{0}(0)  \\
     \boldsymbol{\lambda}(0) 
\end{array}  \right]$denotes the initial mean number of each species. 

\item $\mathcal{M}_{1}$ denotes an $N\times N$ matrix with components
$\eta_{i,j}^{1}$, and $\mathcal{M}_{2}$ denotes an $N\times1$ vector with
components $\eta_{i}^{2}$, where the components are given by:
\begin{equation}%
\eta_{i,j}^{1}=\sum_{m:\bm{y}_{m}=e_{i}}a_{m}(t)(y_{m,j}^{\prime}%
-y_{m,j}),\quad\eta_{i}^{2}=\sum_{m:\bm{y}_{m}=\mathbf{0}}a_{m}%
(t)(y_{m,i}^{\prime}-y_{m,i}),
\end{equation}
in which $e_i$ is a vector where only element $i$ is $1$, and all other elements
are $0$, while $y_{m,i}^{\prime}$ and $y_{m,i}$ are the $i$'th component
of $\bm{y}_{m}^{\prime}$ and $\bm{y}_{m}$, respectively,

\item $\mathcal{N}_{1}$ and $\mathcal{N}_{2}$ denote $N\times N$ matrices with
all elements zero, except the upper left $2\times2$ block, with
\begin{equation}\label{n1n2_def}
\mathcal{N}_{1}=\left(
\begin{array}
[c]{ccc}%
a_{S} & 0 & \cdots\\
0 & -a_{S} & \cdots\\
\vdots & \vdots & \ddots
\end{array}
\right)  ,\quad\mathcal{N}_{2}=\left(
\begin{array}
[c]{ccc}%
\mathrm{i}a_{S} & 0 & \cdots\\
0 & \mathrm{i}a_{S} & \cdots\\
\vdots & \vdots & \ddots
\end{array}
\right),
\end{equation}
where 
$a_S=\sqrt{2a_0}/2$. 
\end{enumerate}

\vspace{0.5em}

\noindent The theorem reads as follows.

\begin{theorem}\label{theo1}

For the system in \eqref{apb_reac}, providing:
\begin{enumerate}
    \item the variables $\boldsymbol{\lambda}$ and $\lambda_{S}$ obey the stochastic 
differential equations (SDEs): 
\begin{equation}
\begin{split}\label{sde_formu1}
{\rm d}\boldsymbol{\lambda}&=\left(\mathcal{M}_1\boldsymbol{\lambda}+\mathcal{M}_2\right){\rm d}t + \mathcal{N}_1\boldsymbol{\lambda}{\rm d}W^1_t+\mathcal{N}_2\boldsymbol{\lambda}{\rm d}W^2_t,\\
{\rm d}\lambda_{S}&= \left( a_S \lambda_1-a_S \lambda_2 \right){\rm d}W^1_t + \left( {\rm i}a_S \lambda_1+{\rm i}a_S \lambda_2 \right){\rm d}W^2_t,
\end{split}
\end{equation}
subject to $\lambda_S(0)=0$, and some initial condition $ \boldsymbol{\lambda}(0)$,
    \item the new variable $\lambda_{0}$ obeys the SDE: 
\begin{equation}
\begin{split}\label{sde_formu2}
{\rm d}\lambda_{0}=a_0\lambda_1\lambda_2{\rm d}t,
\end{split}
\end{equation}
subject to $\lambda_{0}(0)=0$,
\end{enumerate}
then if the initial condition of \eqref{general_cme_element} is a distribution of the Poisson-product form: 
\begin{equation}
\begin{split}\label{sde_ini}
{\rm P}(\textbf{X},0)&=\frac{ \boldsymbol {\Lambda}(0)^{\textbf{X}}}{\textbf{X} !}\exp(-\boldsymbol {\lambda}(0))=\prod_{i=0}^N \frac{ \lambda_{i}(0)^{x_i}}{x_i !} \exp(-\lambda_{i}(0)),
\end{split}
\end{equation}
the solution of \eqref{general_cme_element} is
\begin{equation}
\begin{split}\label{cme_solu}
P(\textbf{X},t)=\left<\frac{\boldsymbol{\Lambda}(t)^{\textbf{X}}}{\textbf{X}!}\exp(-\boldsymbol{\Lambda}(t))\exp(\lambda_0(t)+\lambda_{S}(t))\right>,\\
\end{split}
\end{equation}
 where $<...>$ denotes an expectation operation.

\end{theorem}

Theorem \ref{theo1} is proved in Appendix A. 

Theorem \ref{theo1} gives an exact theoretical expression for the CME solution of a general class of nonlinear biochemical reaction networks with $A + B \rightarrow C$ type of second-order reactions. Compared to the previous analytical solutions of CMEs by Anderson et al.\cite{anderson_time-dependent_2020}, the key feature is that Theorem \ref{theo1} does not require the underlying system to satisfy a dynamically restricted (DR) condition, which considerably constrains the system's structure and parameters. However, the systems we analyze here, as represented by \eqref{general_cme_element}, are not subject to the DR condition. This indicates a broader applicability of Theorem \ref{theo1} than previous exact theoretical results in the literature. 

The key reason why Theorem \ref{theo1} can break the DR condition is that we allow the mean of each molecular species (contained in $\boldsymbol{\lambda}$ and $\lambda_{S}$), to follow stochastic processes, as described by the SDEs of \eqref{sde_formu1}), rather than being smoothly changing continuous variables, with no randomness, whose dynamics are governed by ordinary differential equations (ODEs). As a result, the time-dependent CME solution can be written as the average of a distribution over all \textbf{$\lambda$} paths. More details are given in Appendix \ref{apsec:der}.

Another significant aspect of Theorem \ref{theo1} is its applicability to time-varying systems where the reaction rates are \textit{not} constant. This is a particularly complex problem, as solving CMEs with time-varying rates is much more challenging than those with constant reaction rates\cite{blanes2009magnus,rao2023analytical}. To date, except for our earlier work that provides an exact FPT distribution for a specific type of nonlinear biochemical reaction network with an $A + A \rightarrow C$ type of second-order reactions\cite{rao2023analytical}, we have not encountered any work that offers exact theoretical results for FPTs with $A + B \rightarrow C$ type of second-order reactions and time-varying reaction rates. This suggests that Theorem \ref{theo1} will lead to new avenues of research in this area. 

\subsection{Analytical expression for the FPT distribution}\label{subsub2}

One approach to obtain the solution of Eq. (\ref{cme_solu}) is to calculate the expectation by simulation, i.e., by independently solving the SDEs (for $\boldsymbol{\lambda}$, $\lambda_{S}$ and $\lambda_{0}$) many times and then averaging. This may be very time-consuming. However, Theorem \ref{theo1} can lead to an \textit{exact} 
form of the FPT distribution, as given in \eqref{eq_coro1}, which can then be numerically evaluated.

\begin{corollary}\label{coro1}
The complementary cumulative probability distribution of the FPT
equals the probability of occurrence of states with $x_0=0$:
\begin{equation}\label{eq_coro1}
\begin{split}
P(FPT>t)&= \sum_{\textbf{x}|x_0=0}P(\textbf{X},t)
=\left<\exp(\lambda_{S})\right>.
\end{split}
\end{equation}
\end{corollary}
Corollory \ref{coro1} follows by substituting $P(\textbf{x},x_0,t)$ into \eqref{cme_solu}. 

Equation (\ref{eq_coro1}) is a highly compact exact theoretical result. Next, we present numerical methods for calculating it. 

\vspace{1em}

\subsection{Numerical approximations for the FPT distribution} 

A direct approach to calculating $\left<\exp(\lambda_{S})\right>$ in \eqref{eq_coro1} would require solving an infinite set of coupled ODEs (see Appendix \ref{apsec:why} for details). To avoid such an issue, we shall introduce a numerical approximation method based on a Pad\'e approximant of  $\left<\exp(\lambda_{S})\right>$. This requires the calculation of moments of $\lambda_{S}(t)$. As we can show, the $n$'th moment, $<\lambda_{S}^n>$, has a closed-form expression. This follows because $<\lambda_{S}^n>$ is governed by a finite set of coupled ODEs, which results because  $\boldsymbol{\lambda}$ ( \eqref{sde_formu1}) and $\lambda_{S}$ ( \eqref{sde_formu2}) follow linear SDEs.

We proceed by constructing a function $H(s,t)=\left<\exp(s\lambda_{S}(t))\right>$; the quantity 
required for Eq. (\ref{eq_coro1}) is given by $H(1,t)$. 
To approximate $H(1,t)$, we first obtained a Pad\'e approximant of the function $H(s,t)$, denoted as $T(s,t)$, and then set $s=1$ to calculate $T(1,t)$.  

We determined a Pad\'e approximant of $H(s,t)$ using the following procedure\cite{baker1981morris}: 
\begin{enumerate}
\item 
Construct the Maclaurin series of $H(s,t)$ in $s$, truncated at the $\tilde{N}$'th
term, i.e. $T_{\tilde{N}}(s,t)= \sum_{n=0}^{\tilde{N}} \frac{s^n}{n!} \times \frac{\partial^n}{\partial s^n}H(s,t)|_{s=0} = \sum_{n=0}^{\tilde{N}}b_n(t)s^n$. The larger $\tilde{N}$ is, the better the Pad\'e approximate is, but at the cost of a more time-consuming algorithm. 

\item Calculate $b_n(t)$ by determining $<\lambda_{S}(t)^n>$; since $\frac{\partial^n}{\partial s^n}H(s,t)|_{s=0}= <\lambda_{S}(t)^n>$, $b_n(t)$ is known if $<\lambda_{S}(t)^n>$ is calculated. See 
below for details on how to calculate $<\lambda_{S}(t)^n>$.

\item Find the Pad\'e approximant by determining two polynomials $P^*_{\tilde{L}}(s,t)$ and $Q^*_{\tilde{N}-\tilde{L}}(s,t)$, such that the Maclaurin series of $P^*_{\tilde{L}}(s,t)/Q^*_{\tilde{N}-\tilde{L}}(s,t)$, truncated at the $\tilde{N}$'th term, equals $T_{\tilde{N}}(s,t)$. 

\item Equate the coefficients of $P^*_{\tilde{L}}(s,t)/Q^*_{\tilde{N}-\tilde{L}}(s,t)$ with that of the corresponding polynomial terms of $T_{\tilde{N}}(s,t)$, then $P^*_{\tilde{L}}(s,t)$ and $Q^*_{\tilde{N}-\tilde{L}}(s,t)$  can be uniquely set by solving $\tilde{N}+1$ algebraic equations. For example, if $P^*_{\tilde{L}}(s,t)=p_0(t)+p_1(t)s+p_2(t)s^2+\dots+p_{\tilde{L}}(t)s^{\tilde{L}}$, and %
if
$Q^*_{\tilde{N}-\tilde{L}}(s,t) = 1+q_1(t)s+q_2(t)s^2+\dots+q_{\tilde{N}-\tilde{L}}(t)s^{\tilde{N}-\tilde{L}}$, 
then the algebraic equations are:
\begin{equation}
\begin{split}\label{pade_eq}
b_0=p_0&\\
b_1+b_0q_1=p_1&\\
b_2+b_1q_1+b_0q_2=p_2&\\
\vdots\\
b_{\tilde{L}}+b_{\tilde{L}-1}q_1+\dots+b_0q_{\tilde{L}}=p_{\tilde{L}}&\\
b_{\tilde{L}+1}+b_{\tilde{L}}q_1+\dots+b_0q_{\tilde{L}-1}=0&\\
b_{\tilde{N}}+b_{\tilde{L}}q_1+\dots+b_0q_{2\tilde{L}-\tilde{N}}=0&,\\
\end{split}
\end{equation}
where $q_n=0$ for $n<0$ or $n>\tilde{N}-\tilde{L}$.

\end{enumerate}

\vspace{1em}

The key to obtaining a Pad\'e approximant of $H(s,t)$ is the determination of $<\lambda_{S}(t)^n>$. To achieve this, we adopt the following procedure\cite{press1992flannery}. 

\begin{enumerate}
    \item   
    Differentiate $\lambda_{S}(t)^n$, using Ito's rule: 
${\rm d}(\lambda_{S}^n)=n\lambda_{S}^{n-1}{\rm d}\lambda_{S}+ \frac{n(n-1)}{2}\lambda_{S}^{n-2}({\rm d}\lambda_{S})^2$, and substitute ${\rm d}\lambda_{S}$ and $({\rm d}\lambda_{S})^2$ from the SDEs 
in \eqref{sde_formu1}. The resulting right-hand side is a polynomial in $\lambda_{S}^{l_0}\boldsymbol{\lambda}^{\bm{l}}$. 

\item 
Keep differentiating all new terms of the form $\lambda_{S}^{l_0}\boldsymbol{\lambda}^{\bm{l}}$ and 
substitute ${\rm d}\lambda_{S}$ along with the SDEs in \eqref{sde_formu1}, until all terms in the right-hand side of the equations are known. 

\item Average 
the equations for ${\rm d}(\lambda_{S}^n)$ and ${\rm d}(\lambda_{S}^{l_0}\boldsymbol{\lambda}^{\bm{l}})$; a set of ODEs are obtained, whose solution yields $<\lambda_{S}(t)^n>$. This set of ODEs is of finite dimension as a consequence of linearity of the SDEs in \eqref{sde_formu1}.

\end{enumerate}


\subsection{Theory validation}

We first applied our theory to an exemplar biochemical network that composes four zero/first-order reactions upstream of a second-order reaction:
\begin{equation}\label{exam_sys}
\begin{array}{lll}
\emptyset  &{\xtofrom[\ \;{ a_2(t)}\ ]{\ \;{ a_1(t)}\ }} & S_1 \vspace{1ex},\\
\emptyset  &{\xtofrom[\ \;{ a_4(t)}\ ]{\ \;{ a_3(t)}\ }} & S_2 \vspace{1ex},\\
S_1 + S_2&{\xrightarrow{\ \ \ \;{ a_0(t)}\ \ \ }} &S_0,
\end{array}
\end{equation}

We shall determine the FPT distribution via Eq. (\ref{fpt_aux})

We validated our theoretical results by comparing them to stochastic simulation algorithm (SSA) 
results, which follow from the Gillespie algorithm\cite{gillespie2013perspective,cao2006efficient}. To the best of our knowledge, there are no other exact results for comparison  (our method is the first that provides an exact FPT distribution for biochemical systems with the network structure of \eqref{exam_sys}). Thus, to illustrate the method's effectiveness, we compared our method with another recently developed approximation method, namely, the linear-mapping method (LMA)\cite{cao_linear_2018}, which derives CME solutions of a linearly approximated biochemical network using a mean-field assumption. We chose the LMA as a benchmark because: 1) our method and the LMA method require similar computational time, while moment closure schemes are significantly slower; 2) both our method and the LMA method can be used for biochemical systems with time-varying reaction rates.  

\begin{figure}[ht]
\includegraphics[width=0.95\linewidth]{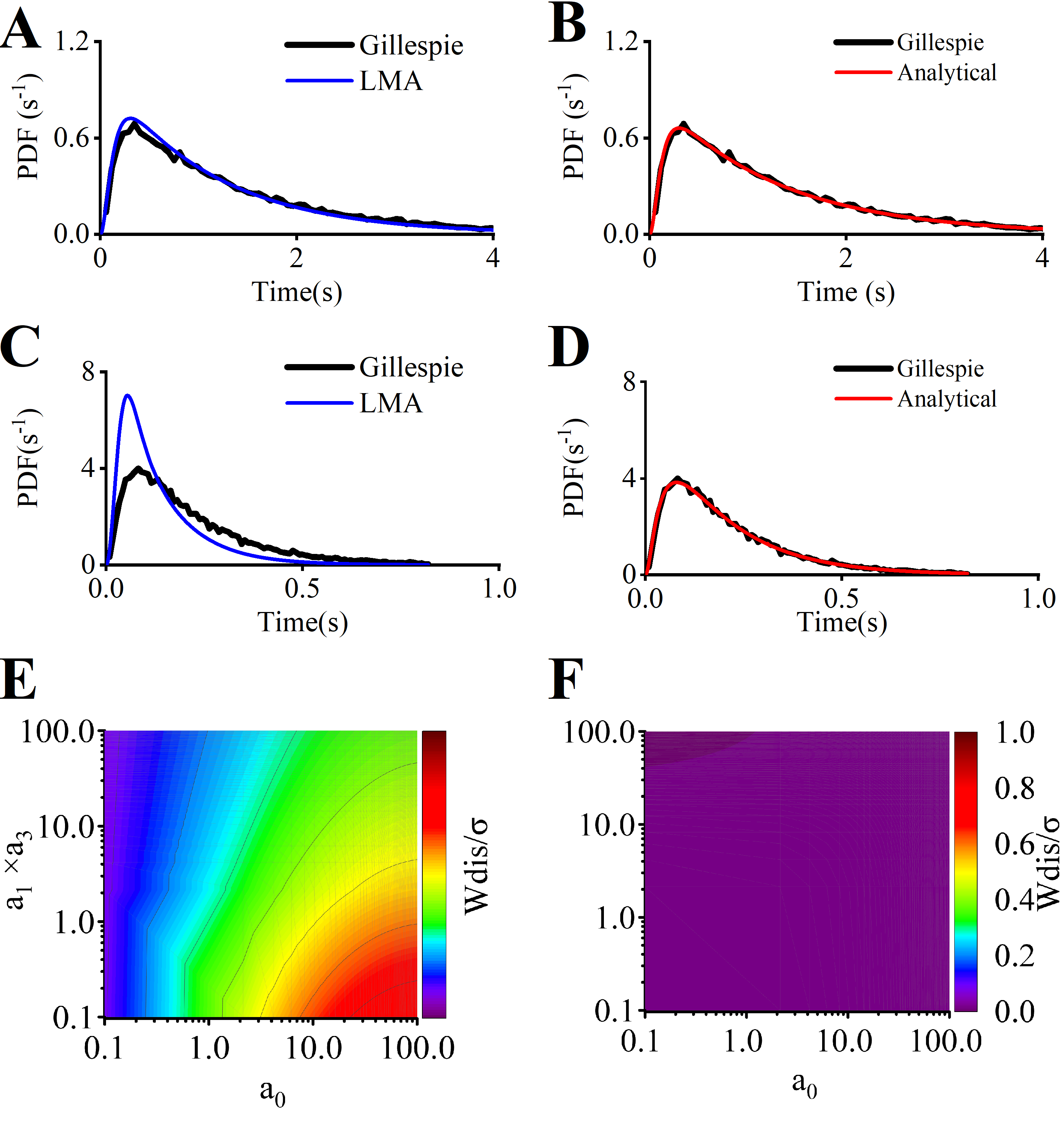}
\caption{\label{fig:tomo_fixedrate} For networks with time-constant reaction rates, the exact method for the FPT distribution, given in the present work, is more accurate and robust than that gained with the LMA method. The LMA (A) and our methods (B) can provide accurate FPT distributions within specific parameter ranges. However, for networks with fast reaction rates or large numbers of reactant molecules in the second-order reaction, our method provides much more accurate FPT distributions (C\&D), and the accuracy is across wide parameter ranges (E\&F). }
\end{figure}

Our method is more accurate and robust than the LMA method, across a broader range of parameters for the network in \eqref{exam_sys} with time-constant reaction rates. The LMA and our method can provide accurate FPT distributions in some parameter ranges (Fig.\ref{fig:tomo_fixedrate} A\&B). However, as the second-order reaction goes faster or with more reactant molecules, our method can be more accurate (Fig.\ref{fig:tomo_fixedrate} D vs. Fig.\ref{fig:tomo_fixedrate} C and Fig.\ref{fig:tomo_fixedrate} F vs. Fig.\ref{fig:tomo_fixedrate} E ). We used the normalized Wasserstein distance (W-distance) to measure the error between the SSA simulated and the analytical FPT distributions gained from the LMA method or our method. The W-distance is normalized against the standard deviation of the SSA-simulated FPT distribution so that the normalized W-distance is dimensionless, allowing comparisons through various time scales. The heat map in Fig. \ref{fig:tomo_fixedrate} E\&F were fitted upon 49 error points, obtained from the results of a ($7 \times 7$) set of parameters. 

We further applied our method to the network in \eqref{exam_sys} with time-varying reaction rates. We restricted consideration to a second-order reaction rate ($a_0$) that varies with time in a sinusoidal fashion. We tested four cases, where $a_0$ can be small (Fig. \ref{fig:time_var} A-B) or large (Fig. \ref{fig:time_var} C-D), or the \textit{rate} of $a_0$ can be fast (Fig. \ref{fig:time_var} A\&C) or slow (Fig. \ref{fig:time_var} B\&D). Again, with larger $a_0$ values, our method is more accurate (Fig. \ref{fig:time_var} C-D), whereas the rate at which $a_0$ changes has a slight influences on accuracy (Fig. \ref{fig:time_var} A\&C vs. Fig. \ref{fig:time_var} B\&D). 

\begin{figure}[ht]
\includegraphics[width=0.95\linewidth]{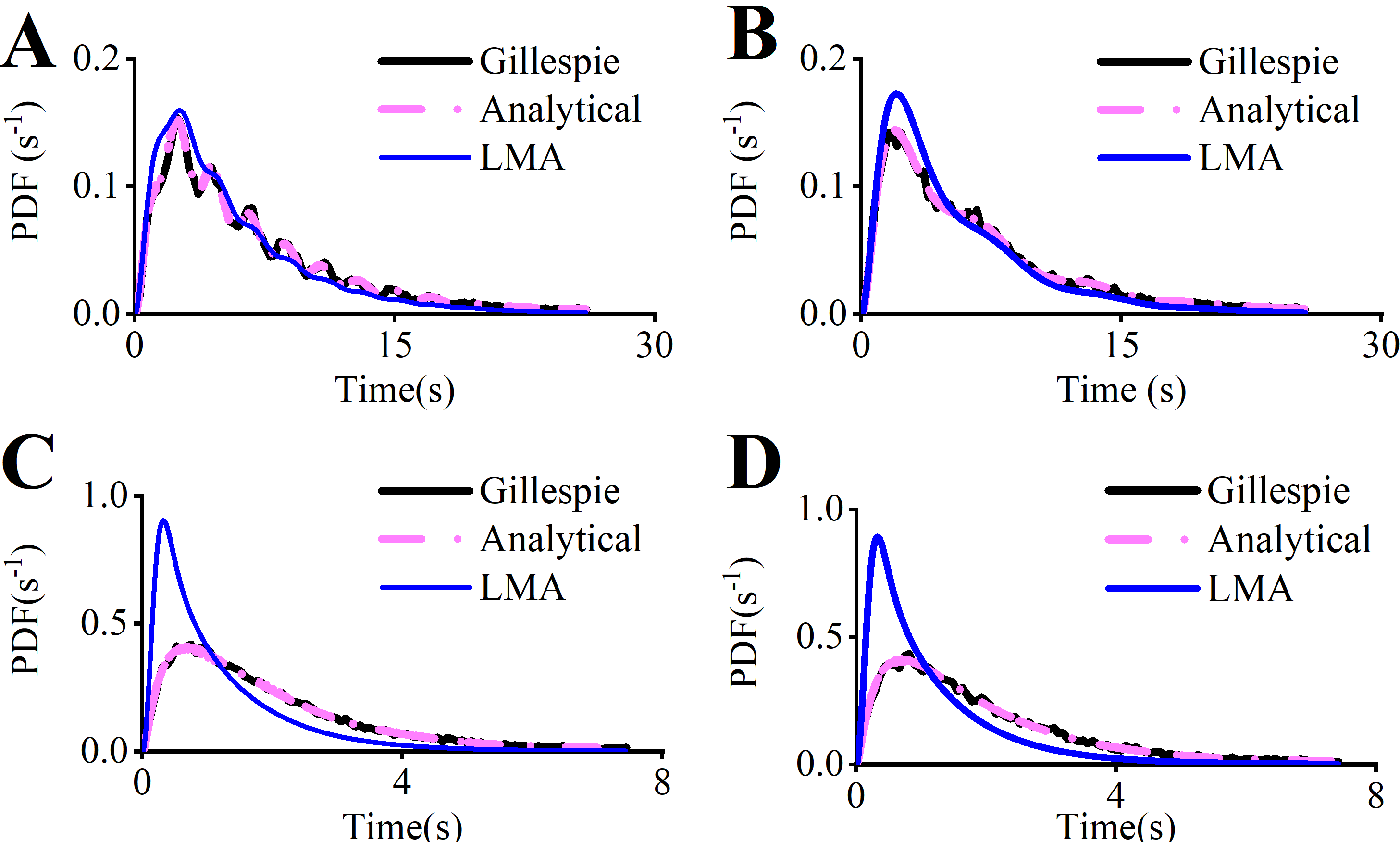}
\caption{\label{fig:time_var} For biochemical networks with time-varying reaction rates, the exact method for the FPT distribution, given in the present work, is more accurate and robust than the LMA method. A. $a_0$ is small, and changes quickly following sin(4t); B. $a_0$ is small, and changes slower following sin(t); C. $a_0$ is large, and changes quicker following sin(4t); D. $a_0$ is large, and changes slower following sin(t).} 
\end{figure}

\subsection{Applications to real biochemical networks}

Our method can be applied to various biochemical systems in different fields. Here, we demonstrate two applications in genetic regulation networks (GRNs)\cite{cao_linear_2018} and a multistep reaction pathway in Ras activation by a protein, called Son Of Sevenless (SOS)\cite{huang2021relating}. 

\vspace{1em}
\noindent \textbf{FPT distribution of a GRN}: \textit{Gene expression} is a fundamental process allowing organisms to create life machinery\cite{chaffey2003alberts}. It is a highly regulated process that involves 
the coordinated action of regulatory proteins that bind to specific DNA sequences to activate or repress gene transcription\cite{orphanides2002unified}. These genetic regulation networks are crucial in cell differentiation, development, and disease\cite{davidson2005gene}. Two decades of research have shown that gene-gene or protein-gene interactions are inherently stochastic, leading to cell-cell variations in mRNA and protein levels\cite{thattai2001intrinsic,swain2002intrinsic}. Thus, analyzing the stochastic timings of the GRNs can be important for understanding cellular phenomena and their function\cite{shahrezaei2008analytical}.  

Here, we analyze a simple GRN system (Fig.\ref{fig:grn} A), which has only one gene that can be in two states, inactive ($G$) and active ($G^*$)\cite{jia2024holimap}. A protein $P$, which is generated and degraded dynamically, 
can bind to $G$ to activate it, and the activated gene, $G^*$, may become deactivated after some time. The GRN involves three chemical reactions, including a second-order reaction, as shown in \eqref{grn}. Starting with some substances that can generate protein $P$ and hence $G^*$, we ask at what time does the protein $P$ activate the gene? 

\begin{equation}
\label{grn}
\begin{array}{lll}
\emptyset&\xtofrom[a2]{a1} &P\\
G^*&{\xrightarrow{\ \ a_3\ \ }} &G,\\
G +P&{\xrightarrow{\ \ a_0\ \ }} &G^*,\\
\end{array}
\end{equation}

To answer this question, we must derive the FPT distribution of the second-order reaction,  in \eqref{grn}. We applied both the LMA method and our method to solve this problem. Our method is accurate across a wide range of parameters (Fig.\ref{fig:grn} B), whereas the LMA can lead to significant errors when the reaction rate of the second-order reaction is high (Fig.\ref{fig:grn} C). 

\begin{figure}[ht]
\includegraphics[width=0.95\linewidth]{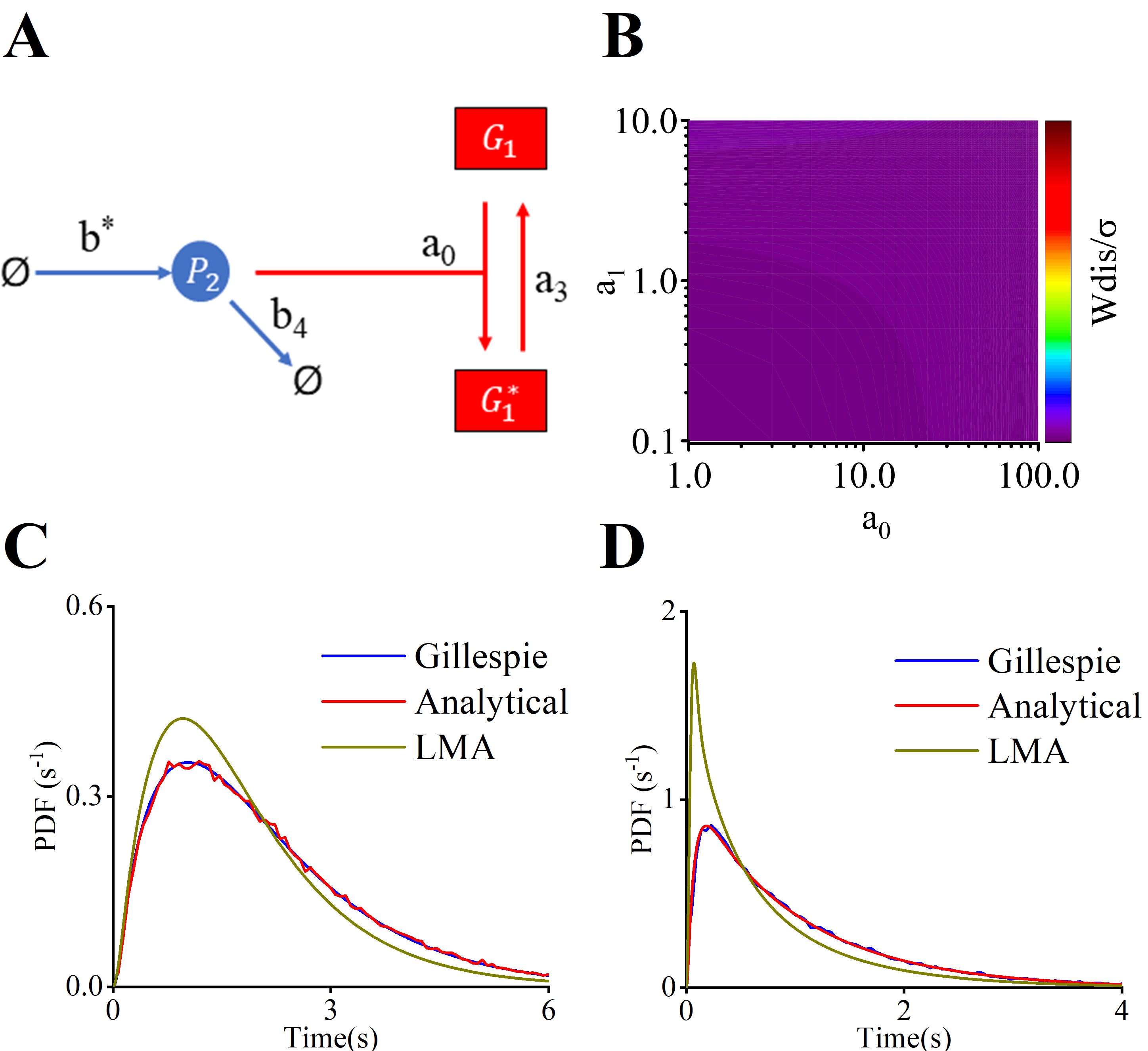}
\caption{\label{fig:grn} For a gene regulatory network (GRN), our exact method derives a much more accurate FPT distribution than the LMA method  A. A simple GRN, where binding a protein $P$ can activate a gene $G$. B. Our exact method accurately derives the FPT distributions across a broad range of parameter sets. The FPT distributions derived from our theoretical method match the SSA-simulated FPT distributions so well that the WD distances between the two stay low across the whole landscape of many parameters. C. The FPT distribution is calculated for a particular parameter set where the second-order reaction is fast ($a_0=1$,$a_3=1$,$b^*=10$, $b_4=1$), showing that our exact method is much more accurate than the LMA method. D. ($a_0=100$,$a_3=1$,$b^*=10$, $b_4=1$)}
\end{figure}

\vspace{1em}

\noindent \textbf{FPT distribution of a multistep reaction pathway}: A \textit{second application} is a multistep reaction pathway, which is present in many biological and chemical processes, such as enzymatic reactions, the 
folding and unfolding of RNA molecules, and the conformation changes of ion channels\cite{zhou2007kinetic}. 

We chose a specific multistep reaction pathway\cite{huang2021relating}, the Ras activation by SOS (Fig. \ref{fig:temp_mul}), to illustrate the effectiveness of our method. SOS is a Ras guanine nucleotide exchange factor (GEF) that plays a central role in numerous cellular signaling pathways, such as the epidermal growth factor receptor and T-cell receptor signaling. 

SOS is autoinhibited in cytosol and activates only after recruitment to the membrane. There are two phases in the activation (Fig. \ref{fig:temp_mul}A) the release of autoinhibition at the membrane through several membrane-mediated intermediates, by a sequence of first-order chemical reactions and 2) the binding of Ras at the allosteric site of SOS by a second-order reaction, which enables the activation of Ras\cite{huang2021relating}. What is the time when Ras is activated, i.e., what is the FPT distribution of the second-order reaction of \eqref{mul}?  

A recent study analyzed the FPT distribution of the first phase of Ras activation\cite{huang2021relating}. It revealed how a faster Ras activation timescale is possible by using much slower activating SOS molecules through multistep reactions. The study resolves the odd discrepancy between the long timescale of individual SOS molecules and the much shorter timescale of Ras activation. More importantly, it demonstrates how rare and early SOS activation events dominate the macroscopic reaction dynamics, implying that the full FPT distribution is required for understanding of this phenomenon, rather than just the mean.

Nevertheless, the analysis presented in\cite{huang2021relating} only focused on first-order multistep reactions. In contrast, our method can include the second-order Ras binding phase of the activation process, pushing the analysis one step further. Without loss of generality, we used two steps in the multistep reaction phase of the SOS activation. The system of reactions is represented in \eqref{mul}, where protein $S_0$ first needs to change to $S_1$ and then $S_*$ before binding with a protein $R$ to activate Ras. The $S_0$ and $S_1$ proteins degrade, with some time constants, and the activated $R_*$ can deactivate to $R$. 

\begin{figure}[ht]
\includegraphics[width=0.95\linewidth]{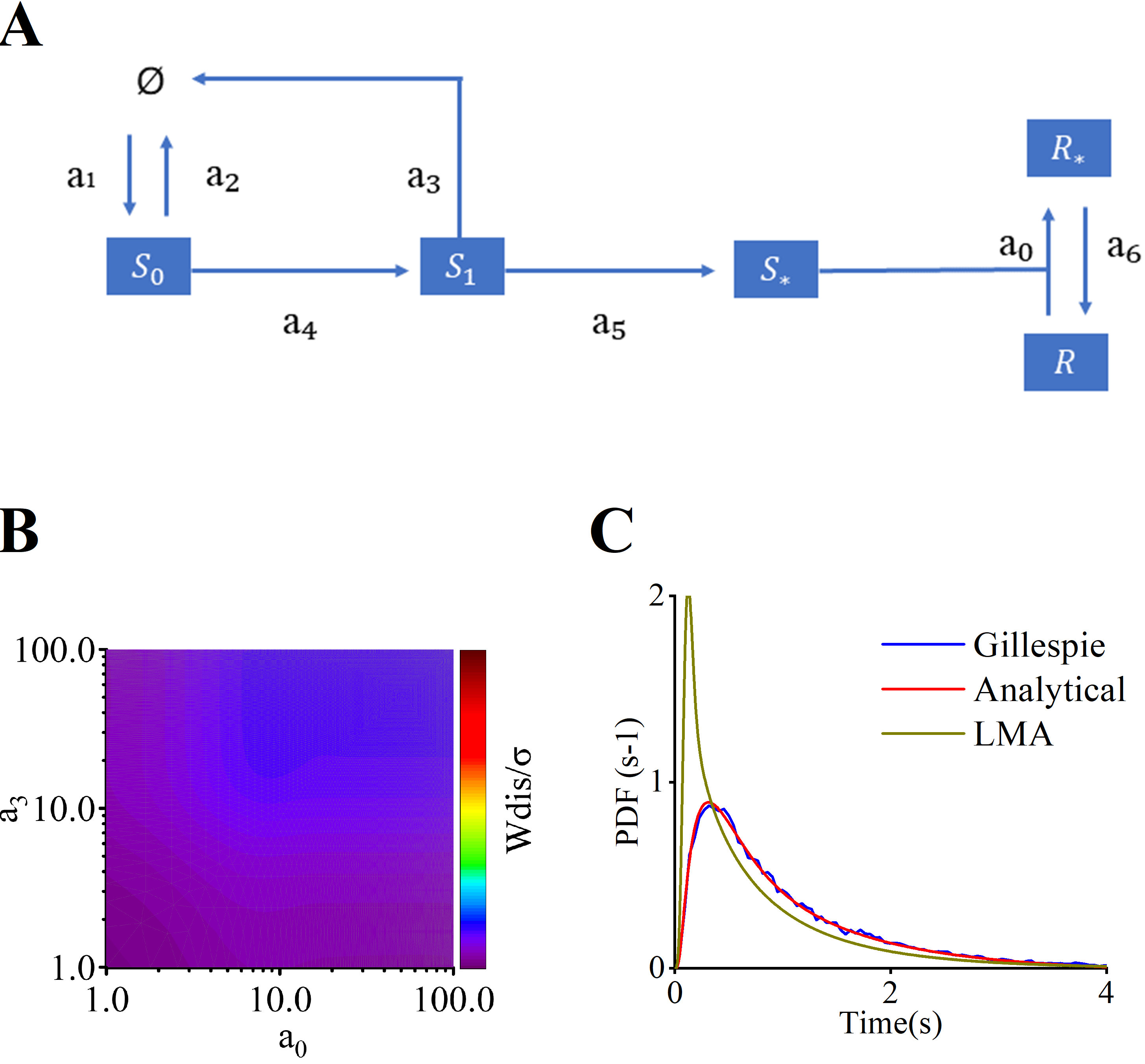}
\caption{Our exact method derives accurate FPT distributions for an exemplar multistep reaction pathway - the Ras activation pathway by SOS. A. A graph illustration of the Ras activation pathway by SOS. B. Our exact method accurately derives the FPT distributions for the Ras activation pathway across various parameter sets. C. The FPT distribution is calculated for a particular parameter set ($a_0=a_3=100$, $a_1=a_2=a_4=10$,$a_5=1$), where the LMA induces a largely inaccurate FPT distribution.
\label{fig:temp_mul} }
\end{figure}

\begin{equation}\label{mul}
\begin{array}{lll}
\emptyset&\xtofrom[a2]{a1} &S_0\\
S_1  &{\xrightarrow{\ \ a_3\ \ }} & \emptyset \vspace{1ex},\\
S_0  &{\xrightarrow{\ \ a_4\ \ }} & S_1 \vspace{1ex},\\
S_1 &{\xrightarrow{\ \ a_5\ \ }} &S_*, \vspace{1ex}\\
R^*&{\xrightarrow{\ \ a_6\ \ }} &R,\\
R+S_2&{\xrightarrow{\ \ a_0\ \ }} &*,\\
\end{array}
\end{equation}

We applied the LMA, along with our methods, to derive the FPT distribution of the binding reaction between $S_*$ and $R$ in \eqref{mul} across a wide range of parameters. We compared the FPT distributions with the SSA-simulated ones and measured the errors by a scaled W-distance. The W-distance is scaled by the standard variable of the SSA-simulated FPT distribution, denoted by $\sigma$. The scale is to make the W-distance a dimensionless quantity, enabling comparisons across systems with different timescales. The errors are all minimal for all parameter sets tested, with the maximum error at a scale of $10\%$ of $\sigma$ (Fig. \ref{fig:temp_mul} B). However, the LMA can induce a largely inaccurate FPT distribution with a W-distance of $90\%$ of $\sigma$ (Fig. \ref{fig:temp_mul} C), 9 times of error compared to our results. 

\section{Discussion}

In this paper we have derived an exact theoretical result for a general class of biochemical networks that involve nonlinearities caused by two-particle collisions. The networks include a second-order reaction of the $A+B \rightarrow C$ type that is downstream of a series of zero or first-order reactions. We have derived the exact theoretical distribution 
for the earliest time the second-order reaction occurs - the first passage time (FPT) distribution of the second-order reaction.

Exact theoretical first passage time distributions have not, previously, been derived for such a system 
with a broad range of applications. This is due to a lack of time-dependent solutions 
of the distribution for the system's states, as described by a chemical master equation (CME). It is known that solutions of the CME, for general biochemical systems, 
only hold for time-constant, 
linear systems composed only of zero and first-order chemical reactions. Even the simplest two-particle collisions cause strong nonlinearities that hinder development of theoretical solutions of the CME for a general system. Additionally, there are further difficulties 
caused by the time-varying reaction rates. 

However, theoretical results for the solution of the CME have advanced to deal with time-varying reaction rates and second-order reactions for \textit{specific systems}. The most recent theoretical results state that for systems that satisfy a dynamically restricted complex balance condition (DR condition), the solution of the CME maintains a Poisson-Product for all times, from the 
initial time. The solution achieves this by deriving time-varying dynamics of the mean number of each molecular species, 
as described by a set of deterministic ordinary differential equations (ODEs). However, the DR condition requires that all reactant pairs are balanced between reactants and products. This is strongly restrictive, and limits
not only to the form of the reactions but also the reaction rates and initial conditions. 

By contrast, the results we have presented in this work represent a solution of the CME in its entirety, and are free of the 
restrictions of the DR condition. This is achieved by 
allowing the mean molecular numbers to be stochastic processes,
rather than having deterministic dynamics; the solution of the CME is obtained by averaging over the stochasticity. 
Based on the exact solution of the CME, we derived the full distribution of the FPT, and developed a numerical scheme, 
based on Pade-approximants, that allow its rapid computation. 

Apart from demonstrating the theoretical value of our results, we have shown its practical value. 
Our method can accurately compute the FPT distributions of nonlinear systems with $A+B \rightarrow C$ type of second-order reactions. The results are much more accurate than state-of-the-art linear mapping approximation methods (LMA), which we specifically chose
as a comparison, because it is a recent method that can be applied to time-varying systems, as can our results. 
We further demonstrated our method's applicability to various networks, including genetic regulation networks and multi-step reaction pathways. This exhibits the potential of our approach to real-world applications. 

While our analysis represents theoretical progress, it is limited to systems that have only one second-order reaction, and what is derived is the FPT of this particular second-order reaction. Therefore, our analysis cannot be applied to biochemical systems that 
do not conform to the description given in \eqref{apb_reac}. One such system is the widely-used enzymatic process given by
the Michaelis Menten reaction, where a first-order reaction follows a reversible second-order reaction. 
We plan to extend our analysis to such systems. 

Overall, our work is a step forward in the theoretical derivation of exact solutions of full FPT distributions for biochemical networks with second-order reactions. Treating the mean molecular numbers as stochastic, and then averaging the result is, we believe, a new approach to theoretical derivations of the solution of CMEs. Indeed, this approach is the key 
that enabled us to derive theoretical solutions of the CME for more general systems, and this may represent
a new way of analyzing more general nonlinear biochemical networks.

\bibliographystyle{plain}
\bibliography{temp}
 
\appendix
\titleformat{\section}{\Large\bfseries}{}{0.3em}{\appendixname~\Alph{section}.}
\titleformat{\subsubsection}{\small\bfseries}{}{0.3em}{\ \ }
\section{Proof of Theorem 1}\label{apsec:der}
In this appendix we give a proof of Theorem 1.

\subsection{General procedure to solve the CME, assuming a solution of Poisson-product form}

For a general second-order biochemical system, as defined in \eqref{apb_reac}, the CME can be written as \eqref{general_cme_element}:  

\begin{equation}
\label{appeq:general_cme_element}
\displaystyle\frac{{\rm d} P(\bm{X}, t )}{{\rm d} t} \! =\!\sum\limits_{k=0}^{M}\left[ P\left(\bm{X}-\bm{v}_k, t \right) c_k\left(\bm{X}-\bm{v}_k,t\right)\!-\! P(\bm{X}, t ) c_k(\bm{X},t)\right],
\end{equation}
where $P(\bm{X}, t )$ is the probability of the system being in state $\bm{X}$ and time $t$. 
The quantity $c_k(\bm{X},t)$ equals the propensity of the $k$'th reaction at time $t$, $a_k(t)\bm{X}^{\bm{Y}_k}$,  where $a_k$ is the rate constant and, for any vectors with $d$ components, a vector to the power of a
vector is defined by $\mathbf{u}^{\mathbf{v}}\overset{\rm def}{=}\prod^d_{i=1}u_i^{v_i}$, with the convention that $0^0=1$.
The quantity $\bm{v}_k$ is the transition vector of the $k$'th reaction.
\setlength{\FrameRule}{1pt}
\setlength{\FrameSep}{10pt}
\begin{framed}
\noindent A key question is: if the initial condition of \eqref{appeq:general_cme_element} has a Poisson product form:
\begin{equation}
\begin{split}\label{appeq:poi_ini}
P(\bm{X},0)&=\prod_{i=0}^N \frac{\lambda_0(0)^{x_i}}{x_i!}\exp(-\lambda_i(0)),
\end{split}
\end{equation}
then under what condition does the solution of \eqref{appeq:general_cme_element} remain of Poisson product form, i.e. 
\begin{equation}
\begin{split}\label{appeq:der_poi_prod}
P(\bm{X},t)= \prod_{i=0}^N \frac{\lambda_i(t)^{x_i}}{x_i!}\exp(-\lambda_i(t)),
\end{split}
\end{equation}
\end{framed}

We define $$\boldsymbol{\lambda}(t)=[\lambda_{1}(t),\dots
,\lambda_{N}(t)]^\top$$ be a column vector containing the mean of all numbers of the species in $\bm{S}$ at
any time $t$ ($t\geq0$), and $$\boldsymbol{\Lambda}(t)= \left[\begin{array}{c}
     \lambda_{0}(t)  \\
     \boldsymbol{\lambda}(t) 
\end{array}  \right]$$ be a column vector containing the mean of all numbers of the species in $\bm{S^*}$ at
any time $t$ ($t\geq0$).

Proceeding by substituting the Poisson product form of \eqref{appeq:der_poi_prod} into 
 \eqref{appeq:general_cme_element}, the left hand side follows by the Chain rule: 

\begin{equation}
\begin{split}\label{appeq:sub_lef}
\displaystyle\frac{{\rm d} P(\bm{X}, t )}{{\rm d} t}=P(\bm{X},t) \times \sum_{i=0}^N \lambda_i'(t) \left(\frac{\bm{X}!}{(\bm{X}-e_i)!}\bm{\Lambda}(t)^{-e_i}-1 \right),
\end{split}
\end{equation}
where $e_{i}$ is a vector within which only element $i$ is $1$, and all other elements are zero.

The right hand side of \eqref{appeq:general_cme_element} is given by:

\begin{equation}\label{appeq:sub_rig}
\begin{split}
 &\sum\limits_{m=0}^{M}\left[ P\left(\bm{X}-\bm{v}_m, t \right) c_m\left(\bm{X}-\bm{v}_m,t\right)\!-\! P(\bm{X}, t ) c_m(\bm{X},t)\right]\\
&= P(\bm{X},t) \left[ \sum_{i=0}^N K_i(t) \left(\frac{\bm{X}!}{(\bm{x}-e_i)!}\bm{\Lambda}(t)^{-e_i}-1 \right)\right.\\& \left.+ \sum_{i=0}^N\sum_{j=0}^N K_{ij}(t) \left( \frac{\bm{X}!}{(\bm{X}-e_i-e_j)!}\bm{\Lambda}(t)^{-e_i-e_j} -1 \right) \right],\\
\end{split}
\end{equation}
where 
\begin{equation}\label{appeq:kij_def}
\begin{split}
K_i(t)&=\sum_{m:\bm{Y}'_m=e_i}a_m(t)\bm{\Lambda}^{\bm{Y}_m}-\sum_{m:\bm{Y}_m=e_i}a_m(t)\bm{\Lambda}^{\bm{Y}_m},\\
K_{ij}(t)&=\sum_{m:\bm{Y}'_m=e_i+e_j}a_m(t)\bm{\Lambda}^{\bm{Y}_m}-\sum_{m:\bm{Y}_m=e_i+e_j}a_m(t)\bm{\Lambda}^{\bm{Y}_m}.
\end{split}
\end{equation}

For the solution of \eqref{appeq:general_cme_element} to remain Poisson-product form over time, we need to equate the right-hand sides of \eqref{appeq:sub_lef} and \eqref{appeq:sub_rig}. The result is:

\begin{equation}
\begin{split}\label{appeq:sub_riglef}
 &\sum_{i=0}^N \lambda_i'(t) \left(\frac{\bm{X}!}{(\bm{X}-e_i)!}\bm{\Lambda}(t)^{-e_i}-1 \right)\\&=\sum_{i=0}^N K_i(t) \left(\frac{\bm{X}!}{(\bm{X}-e_i)!}\bm{\Lambda}(t)^{-e_i}-1 \right) \\&+\sum_{i=0}^N\sum_{j=0}^N K_{ij}(t) \left( \frac{\bm{X}!}{(\bm{X}-e_i-e_j)!}\bm{\Lambda}(t)^{-e_i-e_j} -1 \right).
\end{split}
\end{equation}

In general, this is a complicated nonlinear problem to solve for the $\lambda_i$. In particular, the left-hand side contains only first-order terms of $\bm{\Lambda}$, whereas the right-hand side contains second-order terms of $\bm{\Lambda}$. 

\subsection{Solution that leads to the dynamical and restricted complex balance (DR) condition}

One form of solution of \eqref{appeq:sub_riglef} occurs when all of the $K_{ij}$ coefficients are zero ($K_{ij} = 0$ for all
$i,j$).
Anderson et al. assumed this and found the condition under which this holds (namely the DR condition)\cite{anderson_time-dependent_2020}.

\begin{framed} 

\noindent \textbf{DR condition:} For any higher-order reaction complexes $\bm{z} \cdot \bm{S}$, where the stoichiometric vector $\bm{z}$ satisfies $\bm{z}\in\mathbb{N}^{N+1}_{\geq 0}$ and $\|\bm{z}\|_1\geq 2$, then the DR condition is:
\begin{equation}
\begin{split}\label{appeq:DR_def}
\sum_{m:y_m=z}a_m\bm{\Lambda}^{\bm{z}}(t)=\sum_{m:y'_m=z}a_m\bm{\Lambda}^{\bm{Y}_m}(t),
\end{split}
\end{equation}
where the sum on the left is over those reactions where $\bm{z} \cdot \bm{S}$ are reactants, and the right is over those reactions where $\bm{z} \cdot \bm{S}$ are products. 
\end{framed}

When $K_{ij} = 0$ for all $i$ and $j$ we can write Eq. (\ref{appeq:sub_riglef}) as \newline 
$\sum_{i=1}^{N}(\lambda_{i}^{\prime}(t) -K_{i}(t))\left(  \frac{\bm{x}!}{(\bm{x}-e_{i}%
)!}\bm{\lambda}(t)^{-e_{i}}-1\right) =0$ and a solution arises from 
$\lambda_{i}^{\prime}= K_{i}(t)$.

The dynamics of $\boldsymbol{\Lambda}$ is then governed by a set of deterministic ODEs:
\begin{equation}
\begin{split}\label{appeq:mas_ac}
\frac{\rm d}{\rm d t}\boldsymbol{\Lambda}&=\sum_{m=0}^{M}a_m\boldsymbol{\Lambda}^{\bm{Y}_m}(\bm{Y}_m'-\bm{Y}_m).
\end{split}
\end{equation}

The DR condition says that any higher-order reactant pairs within the system should be balanced. It is very restrictive for the system to satisfy: it requires that the higher-order reactant pairs \textcolor{black}{remain the same} from the reactants to the products and restricts the system's reaction rate coefficients. 

Next, we present our method of solution that does not constrain the system by the DR condition. 

\subsection{{\textcolor{black}{General method of solution  - without requiring the DR condition}}}

In this work we solve \eqref{appeq:sub_riglef} in its entirety, without being restricted by the DR condition. Thus, instead of requiring all $K_{ij}$ to be zero, as applies under a DR condition, we allow 
the sum, that contains the $K_{ij}$, to be present
in \eqref{appeq:sub_riglef}. We incorporate the effects of the sum by 'upgrading' $\bm{\Lambda}$ to be a 
\textit{stochastic process}, which we shall subsequently average over. 

To proceed, we use Ito's rule: $$ {\rm d}f(\lambda)=f'(\lambda){\rm d}\lambda+\frac{1}{2}f''(\lambda)({\rm d}\lambda)^2,$$ when we differentiate the Poisson-product form of \eqref{appeq:sub_lef}, and obtain
\begin{equation}
\begin{split}\label{appeq:subito_lef}
&{\rm d} P(\bm{X}, t )=P(\bm{X},t) \times \left[\sum_{i=0}^N  \left(\frac{\bm{X}!}{(\bm{X}-e_i)!}\bm{\Lambda}(t)^{-e_i}-1 \right){\rm d}\lambda_i(t)\right.\\
&+\left.\frac{1}{2}\sum_{i=0}^N\sum_{j=0}^N  \left(\frac{\bm{X}!}{(\bm{X}-e_i)!}\bm{\Lambda}(t)^{-e_i}-1 \right)\left(\frac{\bm{X}!}{(\bm{X}-e_j)!}\bm{\Lambda}(t)^{-e_j}-1 \right) {\rm d}\lambda_i(t){\rm d}\lambda_j(t)\right].
\end{split}
\end{equation}

\subsubsection*{Determination of the stochastic process, $\bm{\Lambda}$} \ 

Ultimately, we will equate the averaged right-hand side of \eqref{appeq:subito_lef} with 
the right hand side of \eqref{appeq:sub_rig} to determine an SDE that governs the dynamics of the stochastic process, $\bm{\Lambda}$. 


We proceed as follows:  

\begin{enumerate}

\item We assume the following form of the SDE that governs the dynamics of $\boldsymbol{\Lambda}$
(for reasons that will be made clear, shortly):

\begin{equation}
\begin{split}\label{appeq:sde_gen}
{\rm d}\left(\begin{array}{c}
     \boldsymbol{\lambda}  \\
     \lambda_0 
\end{array}\right)=\left(\begin{array}{c}
     \boldsymbol{b}(\boldsymbol{\Lambda},t)  \\
     b_0(\boldsymbol{\Lambda},t) 
\end{array}\right){\rm d}t + \left(\begin{array}{c}
     \bm{\sigma}(\boldsymbol{\Lambda},t)  \\
     \bm{\sigma}_0(\boldsymbol{\Lambda},t) 
\end{array}\right){\rm d}\bm{W}_t,\\
\end{split}
\end{equation}
where $\bm{b}(\boldsymbol{\Lambda},t)$, $b_0(\boldsymbol{\Lambda},t)$ are drift terms, while $\bm{\sigma}(\boldsymbol{\Lambda},t)$ and $\bm{\sigma}_0(\boldsymbol{\Lambda},t)$ are diffusion terms. In the above equation: $\bm{b}(\boldsymbol{\Lambda},t)$ is a N component column vector; $b_0(\boldsymbol{\Lambda},t)$ is a scalar; $\bm{W}_{t}$ is an $2$ component column vector 
of independent Wiener processes; $\bm{\sigma}(\boldsymbol{\Lambda},t)$ is an $N \times 2$ matrix; 
$\bm{\sigma}_0(\boldsymbol{\Lambda},t)$ is a row vector with $2$ elements.

\item We determine $b(\boldsymbol{\Lambda},t)$ and $b_0(\boldsymbol{\Lambda},t)$ by equating them to $K_i(t)$ in \eqref{appeq:sub_rig}, as they are both determined by the zero- and first-order reactions. We obtain:

\begin{equation}
\begin{split}\label{appeq:fz_mas_ac}
&\bm{b}(\boldsymbol{\Lambda},t)=\sum_{m: \|\bm{Y}_m'\|_1\leq 1,\|\bm{Y}_m\|_1\leq 1}a_m\boldsymbol{\lambda}^{\bm{y}_m}(\bm{y}_m'-\bm{y}_m)=\mathcal{M}_1\boldsymbol{\lambda}+\mathcal{M}_2,\\
&b_0(\boldsymbol{\Lambda},t)=a_0\boldsymbol{\lambda}^{\bm{y}_0}(y'_{0,0}-y_{0,0})=a_0\lambda_1\lambda_2.
\end{split}
\end{equation}

\noindent Here $\mathcal{M}_{1}$ is an $N\times N$ matrix with components
$\eta_{i,j}^{1}$, $\mathcal{M}_{2}$ is an $N$ component column vector with
components $\eta_{i}^{2}$, while $y_{m,i}^{\prime}$ and $y_{m,i}$ are the $i$'th
component of $\bm{y}_{m}^{\prime}$ and $\bm{y}_{m}$, respectively, and

\begin{equation}%
\begin{split}
&\eta_{i,j}^{1}=\sum_{m:\bm{y}_{m}=e_{i}}a_{m}(t)(y_{m,j}^{\prime}%
-y_{m,j}),\\&\eta_{i}^{2}=\sum_{m:\bm{y}_{m}=\mathbf{0}}a_{m}%
(t)(y_{m,i}^{\prime}-y_{m,i}).
\end{split}
\end{equation}

\item 
We now explicitly write $\bm{W}_t$ as $\bm{W}_t=[W^1_t,W^2_t]^\top$, then
$\bm{\sigma}(\boldsymbol{\Lambda},t) \cdot {\rm d}\bm{W}_t$ can then be written as: 
\begin{equation}\label{appeq:lam_dif}
\begin{split}
\bm{\sigma}(\boldsymbol{\Lambda},t) \cdot {\rm d}\bm{W}_t& = \sigma_1(\boldsymbol{\Lambda},t){\rm d}W^1_t+ \sigma_2(\boldsymbol{\Lambda},t){\rm d}W^2_t
\end{split}
\end{equation}
where $ \sigma_1(\boldsymbol{\Lambda},t)$ and $ \sigma_2(\boldsymbol{\Lambda},t)$ are $N$ component column vectors.

\noindent We adopted this particular format because the only second-order reactions that occur
within the system are of the type
$A + B\rightarrow C$, with two different reactants. The form of \eqref{appeq:lam_dif} ensures that $K_{ij}$ in \eqref{appeq:sub_rig} satisfies: 

\begin{enumerate} 
\item $K_{ij} = 0$ if $i>2$ or $j>2$, which means that species $S_k$ for $k=3,4\dots$ are not involved in the second-order reaction.

\item $K_{ij} = 0$ if $i=0$ or $j=0$, which means that $S_0$ is not a reactant in any one reaction. It follows that \begin{equation}
    \bm{\sigma}_0(\bm{\Lambda},t)=\bm{0}.
\end{equation}

\item $K_{ij} = 0$ if $i=j$, which means that there are no square terms of the form $\frac{\bm{x}!}{(\bm{x}-2e_i)!}\bm{\lambda}(t)^{-2e_i}$ in the right-hand side of \eqref{appeq:sub_rig}.
\\

When  $\sigma_1(\boldsymbol{\Lambda},t)$ and $\sigma_2(\boldsymbol{\Lambda},t)$ satisfy certain relationships, the non-zero square terms of $\bm{\lambda}$ can be eliminated in the right-hand side of \eqref{appeq:subito_lef} (see below).
\\
\item $K_{12} = K_{21} = -\frac{a_0}{2}\lambda_1\lambda_2$, which means that the matrix $\bm{K} = \{K_{ij}\}$ is symmetric, and its components are all products of different components of $\bm{\Lambda}$. To satisfy this constraint, $\sigma_1(\boldsymbol{\Lambda},t)$ and $\sigma_2(\boldsymbol{\Lambda},t)$ are set as linear functions of $\bm{\lambda}$:

\begin{equation}
\begin{split}\label{appeq:sigma_solu}
&\sigma_1(\boldsymbol{\Lambda},t)= \mathcal{N}_1\boldsymbol{\lambda}\\
&\sigma_2(\boldsymbol{\Lambda},t)=\mathcal{N}_2\boldsymbol{\lambda},\\
\end{split}
\end{equation}
where $\mathcal{N}_1$ and $\mathcal{N}_2$ are $N \times N$ square matrices. 

\end{enumerate}

\vspace{1em}

\item We determine the matrices $\mathcal{N}_1$ and $\mathcal{N}_2$, by equating terms of the form $\frac{\bm{X}!}{(\bm{X}-e_i-e_j)!}\bm{\Lambda}(t)^{-e_i-e_j}$ in \eqref{appeq:subito_lef} with the $K_{ij}$ terms in \eqref{appeq:sub_rig}. This leads to an algebraic equation, the solution of which sets the forms of $\mathcal{N}_1$ and $\mathcal{N}_2$: 
\begin{equation}
    \frac{1}{2}\left(\sigma_1(\boldsymbol{\Lambda},t)\cdot \sigma_1(\boldsymbol{\Lambda},t)^\top+\sigma_2(\boldsymbol{\Lambda},t)\cdot \sigma_2(\boldsymbol{\Lambda},t)^\top\right)=\bm{K},
\end{equation}
where $\bm{K}$ is the matrix with elements $K_{ij}$ ($i\geq 1$, $j\geq 1$) in \eqref{appeq:sub_rig}. $\mathcal{N}_1$ and $\mathcal{N}_2$ are determined as that shown in \eqref{n1n2_def}.

\end{enumerate}

At this point, the dynamics of $\lambda_0$ and $\bm{\lambda}$ are determined by: 

\begin{equation}\label{appeq:bmlambda_def}
\begin{split}
    &{\rm d}\boldsymbol{\lambda}=\left(\mathcal{M}_1\boldsymbol{\lambda}+\mathcal{M}_2\right){\rm d}t + \mathcal{N}_1\boldsymbol{\lambda}{\rm d}W^1_t+\mathcal{N}_2\boldsymbol{\lambda}{\rm d}W^2_t,\\
&{\rm d}\lambda_0=\lambda_1\lambda_2{\rm d}t.
\end{split}
\end{equation}
However, on substituting \eqref{appeq:bmlambda_def} into \eqref{appeq:subito_lef}, terms are generated in the right-hand side of \eqref{appeq:subito_lef} but are absent in \eqref{appeq:sub_rig}. These terms are $\lambda_1\lambda_2\frac{\bm{x}!}{(\bm{x}-e_1)!}\bm{\lambda}(t)^{-e_1}$ and $\lambda_1\lambda_2\frac{\bm{x}!}{(\bm{x}-e_2)!}\bm{\lambda}(t)^{-e_2}$. 

To eliminate these terms in the right-hand side of \eqref{appeq:subito_lef}. We modified the Poisson-product form to $\tilde{P}(\textbf{X},t)$,
as given by

\begin{equation}
\begin{split}\label{appeq:prod_mod}
\tilde{P}(\textbf{X},t)=\frac{\boldsymbol{\Lambda}(t)^{\textbf{X}}}{\textbf{X}!}\exp(-\boldsymbol{\Lambda}(t))\exp(\lambda_0(t)+\lambda_{S}(t)),\\
\end{split}
\end{equation}
where $\lambda_S(t)$ is a stochastic process to be determined.

Note that $\tilde{P}(\textbf{X},t)$ is \textit{not} normalized to unity, and thus is not a probability distribution. 
However, we will prove later that its average, $<\tilde{P}(\textbf{X},t)>$, is a distribution; it is the solution of the CME 
and obeys the Poisson-product form initial condition of \eqref{appeq:poi_ini}.

\subsubsection*{Determination of $\lambda_S$} \ 


Changing the Poisson product form from $P(\bm{X},t)$ to $\tilde{P}(\bm{X},t)$ does not change the format of \eqref{appeq:sub_rig}:

\begin{equation}\label{appeq:sub_rig1}
\begin{split}
&\sum\limits_{m=0}^{M}\left[ \tilde{P}\left(\bm{X}-\bm{v}_m, t \right) c_m\left(\bm{X}-\bm{v}_m,t\right)\!-\! \tilde{P}(\bm{X}, t ) c_m(\bm{X},t)\right]{\rm d}t\\
 &=\tilde{P}(\bm{X},t) \left[ \sum_{i=0}^N K_i(t) \left(\frac{\bm{X}!}{(\bm{X}-e_i)!}\bm{\Lambda}(t)^{-e_i}-1 \right)\right. \\&\left.+ \sum_{i=0}^N\sum_{j=0}^N K_{ij}(t) \left( \frac{\bm{X}!}{(\bm{X}-e_i-e_j)!}\bm{\Lambda}(t)^{-e_i-e_j} -1 \right) \right]{\rm d}t\\
\end{split}
\end{equation}
where $K_i(t)$ is the drift term of ${\rm d}\lambda_i$.

However, when we differentiate $\tilde{P}(\bm{X},t)$ using Ito's rule, \eqref{appeq:subito_lef} is changed to a more complex format, as in \eqref{appeq:subito_dec1}.

Setting $I_{e}=\exp(-\bm{\Lambda})$ and $I_{p}=\frac{\bm{\Lambda}^{\bm{X}}}{\bm{X}!}$ allows us to write
$\tilde{P}(\bm{X},t)= I_{e}I_{p}\exp(\lambda_0(t)+\lambda_{S}(t))$. The quantity ${\rm d} \tilde{P}(\bm{X}, t )$ then becomes:

\begin{equation}
\begin{split}\label{appeq:subito_dec1}
&{\rm d} \tilde{P}(\bm{X}, t )=\left[I_{p}{\rm d}\left(I_{e}\exp(\lambda_0(t)+\lambda_{S}(t))\right)+I_{e}\exp(\lambda_0(t)+\lambda_{S}(t)){\rm d}I_{p}\right]+\exp(\lambda_0(t)){\rm d}I_{p}{\rm d}\left(I_{e}\exp(\lambda_{S}(t))\right)\\
&=\tilde{P}(\bm{X},t) \times \left[\sum_{i=0}^N  K_i(t)\left(\frac{\bm{X}!}{(\bm{X}-e_i)!}\bm{\Lambda}(t)^{-e_i}-1\right){\rm d}t\right.\\
&+\sum_{i=0}^N\sum_{j=0}^N K_{ij}(t) \left( \frac{\bm{X}!}{(\bm{X}-e_i-e_j)!}\bm{\Lambda}(t)^{-e_i-e_j} -1 \right){\rm d}t\\
&+a_S\left[x_1-x_2\right]{\rm d}W^1_t+{\rm i}a_S\left[x_1+x_2\right]{\rm d}W^2_t\\
&-\left( a_S \lambda_1-a_S \lambda_2 \right){\rm d}W^1_t - \left( {\rm i}a_S \lambda_1+{\rm i}a_S \lambda_2 \right){\rm d}W^2_t\\
&\left.+{\rm d}\lambda_{S}+\frac{1}{2}({\rm d}\lambda_{S})^2+\sum_{i=1}^2{\rm d}\lambda_{i}{\rm d}\lambda_{S}-a_0\lambda_1\lambda_2{\rm d}t\right]\\
&+\exp(\lambda_0(t)){\rm d}I_{p}{\rm d}\left(I_{e}\exp(\lambda_{S}(t))\right),
\end{split}
\end{equation}
where
\begin{equation}
\begin{split}\label{appeq:subito_dec4}
&{\rm d}I_{p}{\rm d}\left(I_{e}\exp(\lambda_{S}(t))\right)
=\tilde{P}(\bm{X},t) \times \left[a_0  \left(\lambda_2x_1+\lambda_1x_2 \right) {\rm d}t\right.\\&\left.
+\frac{1}{2}\left(\frac{x_1}{\lambda_1}{\rm d}\lambda_1(t){\rm d}\lambda_S(t)+\frac{x_2}{\lambda_2}{\rm d}\lambda_2(t){\rm d}\lambda_S(t) \right) 
\right].
\end{split}
\end{equation}

We then determine $\lambda_{S}$ by equating the coefficients of the ${\rm d}t$ terms on the right-hand sides of \eqref{appeq:sub_rig1} and \eqref{appeq:subito_dec1}. This leads to constraints that $\lambda_S$ has to satisfy: 

\begin{enumerate}
    \item ${\rm d}\lambda_S$ is a function of ${\rm d} \bm{W}_t$, as \eqref{appeq:sub_rig1} and \eqref{appeq:subito_dec1}  already have 
    corresponding ${\rm d} t$ terms

    \item setting ${\rm d}\lambda_S = \xi_1(\bm{\lambda},t){\rm d}W^1_t + \xi_2(\bm{\lambda},t){\rm d}W^2_t$ 
    allows ${\rm d}I_{p}{\rm d}\left(I_{e}\exp(\lambda_{S}(t))\right)$ to be zero, i.e., the right hand side of \eqref{appeq:subito_dec4} is zero: 
    
\begin{equation}\label{appeq:ls_cons1}
     a_0  \left(\lambda_2x_1+\lambda_1x_2 \right) {\rm d}t
+\frac{1}{2}\left(\frac{x_1}{\lambda_1}{\rm d}\lambda_1(t){\rm d}\lambda_S(t)+\frac{x_2}{\lambda_2}{\rm d}\lambda_2(t){\rm d}\lambda_S(t) \right) =0  
    \end{equation}

    \item ${\rm d}\lambda_S$ should be a quadratic function of ${\rm d}\lambda_S$:

\begin{equation}\label{appeq:ls_cons2}
    \begin{split}
        \frac{1}{2}({\rm d}\lambda_{S})^2+\sum_{i=1}^2{\rm d}\lambda_{i}{\rm d}\lambda_{S}-a_0\lambda_1\lambda_2{\rm d}t=0.
    \end{split}
\end{equation} 

\end{enumerate}

 \eqref{appeq:ls_cons1} and \eqref{appeq:ls_cons2} determine that the solution of ${\rm d}\lambda_S$ is the sum of the diffusion terms of ${\rm d}\lambda_1$ and ${\rm d}\lambda_2$:

\begin{equation}
\begin{split}\label{appeq:sde_formu2}
{\rm d}\lambda_{S}&= \left( a_S \lambda_1-a_S \lambda_2 \right){\rm d}W^1_t + \left( {\rm i}a_S \lambda_1+{\rm i}a_S \lambda_2 \right){\rm d}W^2_t.
\end{split}
\end{equation}

The introduction of $\lambda_{S}(t)$ allows the elimination of the cross-product terms in \eqref{appeq:subito_dec4}, 
corresponding to terms of the form $ \lambda_1(t)\lambda_2(t)\frac{\bm{X}!}{(\bm{X}-e_1)!}\bm{\Lambda}(t)^{-e_1} $ and $ \lambda_1(t)\lambda_2(t)\frac{\bm{X}!}{(\bm{X}-e_2)!}\bm{\Lambda}(t)^{-e_2} $ in \eqref{appeq:subito_lef}. 

We now substitute $\lambda_S$ into \eqref{appeq:subito_dec1} to obtain 

\begin{equation}
\begin{split}\label{appeq:subito_dec2}
&{\rm d} \tilde{P}(\bm{X}, t )=\tilde{P}(\bm{X},t) \times \left[\sum_{i=0}^N  K_i(t)\left(\frac{\bm{X}!}{(\bm{X}-e_i)!}\bm{\Lambda}(t)^{-e_i}-1\right){\rm d}t\right.\\
&+\sum_{i=0}^N\sum_{j=0}^N K_{ij}(t) \left( \frac{\bm{X}!}{(\bm{X}-e_i-e_j)!}\bm{\Lambda}(t)^{-e_i-e_j} -1 \right){\rm d}t\\
&+\left.a_S\left[x_1-x_2\right]{\rm d}W^1_t+{\rm i}a_S\left[x_1+x_2\right]{\rm d}W^2_t\right]\\
\end{split}
\end{equation}

To prove Theorem \ref{theo1}, we needed to equate the right-hand sides of \eqref{appeq:sub_rig1} and \eqref{appeq:subito_dec2}, which means that the stochastic diffusion terms involving ${\rm d}W_t^1$ and ${\rm d}W_t^2$ should be eliminated. Next, we prove that the average of $\tilde{P}(\bm{X},t)$ eliminates these stochastic terms, and $<\tilde{P}(\bm{X},t)>$ satisfies \eqref{appeq:sub_rig1}, and hence is the solution of the CME in \eqref{appeq:general_cme_element}.

\subsubsection*{Eliminating stochastic diffusion terms in \eqref{appeq:subito_dec2}}\ 

Averaging both sides of \eqref{appeq:subito_dec2} yields

\begin{equation}
\begin{split}\label{appeq:subito_dec2_average}
& {\rm d}<\tilde{P}(\bm{X}, t )> =\left<\tilde{P}(\bm{X},t) \times \left[\sum_{i=0}^N  K_i(t)\left(\frac{\bm{X}!}{(\bm{X}-e_i)!}\bm{\Lambda}(t)^{-e_i}-1\right){\rm d}t\right.\right.\\
&+\left.\left.\sum_{i=0}^N\sum_{j=0}^N K_{ij}(t) \left( \frac{\bm{X}!}{(\bm{X}-e_i-e_j)!}\bm{\Lambda}(t)^{-e_i-e_j} -1 \right){\rm d}t \right]\right> \\
&+\left<a_S\tilde{P}(\bm{X},t)\left[x_1-x_2\right]{\rm d}W^1_t+{\rm i}a_S\tilde{P}(\bm{X},t)\left[x_1+x_2\right]{\rm d}W^2_t\right>
\end{split}
\end{equation}

Set $\tilde{\sigma}_1(\bm{X},\bm{\lambda},t)$  = $a_S\tilde{P}(\bm{X},t)\left[x_1-x_2\right]$  and $\tilde{\sigma}_2(\bm{X},\bm{\lambda},t)$ = ${\rm i}a_S\tilde{P}(\bm{X},t)\left[x_1+x_2\right]$, we then prove that: $\left<\tilde{\sigma}_1(\bm{X},\bm{\lambda},\tau){\rm d}W^1_{\tau}+\tilde{\sigma}_2(\bm{X},\bm{\lambda},\tau){\rm d}W^2_{\tau}\right> = 0$ (see Proposition \ref{prop:tlim}), and hence the stochastic diffusion terms in \eqref{appeq:subito_dec2_average} are eliminated.  

\begin{proposition}\label{prop:tlim}
Define $t_n$ be a series of stopping times $t_n={\rm inf}\{\tau\,|\,\|\boldsymbol{\lambda}(\tau)\|\geq n\}$ for $n=1,2,\dots$.

Then, for any $t\in[0,t_n]$, the following expectation is zero:
    \begin{equation}\label{appeq:dif_eli}
    \begin{split}       \left<\tilde{\sigma}_1(\bm{X},\bm{\lambda},t){\rm d}W^1_{t}+\tilde{\sigma}_2(\bm{X},\bm{\lambda},t){\rm d}W^2_{t}\right>=0,
    \end{split}
\end{equation}
\end{proposition}

\begin{proof} \textbf{\ref{prop:tlim}}:

From the definition of $t_n$, when $t\in[0,t_n]$, $\|\boldsymbol{\lambda}(\tau)\|$ is bounded, and hence   $|\tilde{\sigma}_1(\bm{X},\bm{\lambda},t)|$ and $|\tilde{\sigma}_2(\bm{X},\bm{\lambda},t)|$ are bounded.

Thus, $Z(t)=\int^t_0\tilde{\sigma}_1(\bm{X},\bm{\lambda},\tau){\rm d}W^1_{\tau}+\tilde{\sigma}_2(\bm{X},\bm{\lambda},\tau){\rm d}W^2_{\tau}$ is a martingale, and according to the properties of martingales, $E(Z(t)) = Z(0)= 0$, which is in equivalent to \eqref{appeq:dif_eli}.

\end{proof}

Proposition \ref{prop:tlim} guarantees that the stochastic diffusion term is eliminated when $t<t_n$. Next, we prove that $<\tilde{P}(\bm{X}, t )> $ is the solution of the CME in \eqref{appeq:general_cme_element}.

\vspace{1em}

\subsubsection*{$<\tilde{P}(\bm{X}, t )> $ is the solution of the CME in \eqref{appeq:general_cme_element}}

We first prove that $<\tilde{P}(\bm{X}, t )> $ is the solution of the CME in \eqref{appeq:general_cme_element} when $t<t_n$, then we generalise the result as $t_n\rightarrow\infty$.

When $ t< t_n$, \eqref{appeq:dif_eli} holds, and $\left<\tilde{P}(\bm{X}, t) \right>$ is the solution of the CME \eqref{appeq:general_cme_element}, because \eqref{appeq:subito_dec2_average} becomes:

\begin{equation}\label{appeq:t_leq_tn_cme}
\begin{split}
{\rm d}<\tilde{P}(\bm{X},t) > = 
   &\left<\tilde{P}(\bm{X},t) \times \left[\sum_{i=0}^N  K_i(t)\left(\frac{\bm{X}!}{(\bm{X}-e_i)!}\bm{\Lambda}(t)^{-e_i}-1\right){\rm d}t\right.\right.\\
&+\left.\left.\sum_{i=0}^N\sum_{j=0}^N K_{ij}(t) \left( \frac{\bm{X}!}{(\bm{X}-e_i-e_j)!}\bm{\Lambda}(t)^{-e_i-e_j} -1 \right){\rm d}t \right]\right> \\
&=\sum\limits_{m=0}^{M}\left[ \left<\tilde{P}\left(\bm{X}-\bm{v}_m, t \right)\right> c_m\left(\bm{X}-\bm{v}_m,t\right)\!-\! \left<\tilde{P}(\bm{X}, t )\right> c_m(\bm{X},t)\right]{\rm d}t\\ 
\end{split}
\end{equation}

Define $t\wedge t_n=\min(t,t_n)$, and denote $\mathcal{P}(\bm{X},t)$ as the solution of the CME in \eqref{appeq:general_cme_element}. \eqref{appeq:t_leq_tn_cme} means that $\mathcal{P}(\bm{X},t\wedge t_n)=\left<\tilde{P}(\bm{X}, t\wedge t_n ) \right>$ at time $ t\wedge t_n$. 

Next, we prove that when $n\rightarrow\infty$, we have the limit $\lim\limits_{n\rightarrow \infty} t\wedge t_n =t$, leading to

\begin{equation}\label{appeq:plim}
\begin{split}
<\tilde{P}(\textbf{X},t)>=\lim\limits_{n\rightarrow \infty}<\tilde{P}(\textbf{X},t\wedge t_n)>=\lim\limits_{n\rightarrow \infty}\mathcal{P}(\textbf{X},t\wedge t_n)=\mathcal{P}(\bm{X},t),
\end{split}
\end{equation}
where the last equal sign stems from $\mathcal{P}(\bm{X},t)$, as the solution of CME \eqref{appeq:general_cme_element}, being a continuous function.

By probability measure, the limit $\lim\limits_{n\rightarrow \infty}t\wedge t_n=t$ is equivalent to 
\begin{equation}\label{appeq:t_leq_tn_pro}
    \lim\limits_{n\rightarrow \infty}{\rm P}\left(t_n\leq t\right)\rightarrow 0.
\end{equation}

From the definition $t_n={\rm inf}\{t\,|\,\|\boldsymbol{\lambda}(t)\|\geq n\}$, \eqref{appeq:t_leq_tn_pro} is equivalent to 
\begin{equation}\label{appeq:t_leq_tn_lamn}
    \lim\limits_{n\rightarrow \infty}{\rm P}\left(\max_{0\leq \tau\leq t}\|\boldsymbol{\lambda}(\tau)\|\geq n\right)\rightarrow 0.
\end{equation}

Next, we use Lemma \ref{app:lem1} to prove \eqref{appeq:t_leq_tn_lamn}.

The linear SDE in \eqref{appeq:bmlambda_def}, that  governs $\boldsymbol{\lambda}$, satisfies the condition of Lemma \ref{app:lem1}, i.e, \eqref{appeq:lem1_cond} because:

1) \eqref{appeq:bmlambda_def} can be rewritten as
\begin{equation}
    {\rm d}\boldsymbol{\lambda}=\bm{b}(\boldsymbol{\Lambda},t){\rm d}t + \bm{\sigma}(\boldsymbol{\Lambda},t){\rm d}\bm{W}_t,
\end{equation}
where $\bm{b}$ is an $N$ column vector of linear functions defined in \eqref{appeq:fz_mas_ac}, and $\bm{\sigma}$  is an 
$N\times 2$ matrix of linear functions defined in \eqref{appeq:lam_dif}.

2) Linearity of $\bm{b}$ and $\bm{\sigma}$ lead to $\| \bm{b}(t,y)\|^2+\| \bm{\sigma}(t,y)\|^2$ being quadratic functions of the vector 
$y\in\mathbb{R}^N$. Thus, \eqref{appeq:lem1_cond} holds, which means

\begin{equation}\label{appeq:lemm1_corr}
    \left<\max_{0\leq \tau\leq t}\|\boldsymbol{\lambda}(\tau)\|^2\right> \leq K(\boldsymbol{\lambda}(0)){\rm e}^{Ct},
\end{equation}

where $C$ is a constant, and $K(\boldsymbol{\lambda}(0))$ is a constant that depends on $\boldsymbol{\lambda}(0)$.

Then by Chebyshev's inequality
\begin{equation}\label{appeq:lemm1_corr1}
\begin{split}
    &{\rm P}\left(\max_{0\leq \tau\leq t}\|\boldsymbol{\lambda}(\tau)\|\geq n\right)\leq {\rm P}\left(\left|\max_{0\leq \tau\leq t}\|\boldsymbol{\lambda}(\tau)\|- \mu\right|\geq n-\mu\right)\\&\leq \frac{\left<\displaystyle\max_{0\leq \tau\leq t}\|\boldsymbol{\lambda}(t)\|^2\right>}{(n-\mu)^2}\leq \frac{K(\boldsymbol{\lambda}(0)){\rm e}^{Ct}}{(n-\mu)^2},  
\end{split} 
\end{equation}
where $\mu=\left<\|\boldsymbol{\lambda}(t_{m})\|\right>$ and $\|\boldsymbol{\lambda}(t_{m})\|=\max_{0\leq \tau\leq t}\|\boldsymbol{\lambda}(\tau)\|$.

Because 
\begin{equation}\label{appeq:lemm1_corr2}
    \mu^2=\left<\|\boldsymbol{\lambda}(t_{m})\|\right>^2\leq \left<\|\boldsymbol{\lambda}(t_{m})\|^2\right>\leq K(\boldsymbol{\lambda}(0)){\rm e}^{Ct},
\end{equation}
we have the limit
\begin{equation}\label{appeq:lemm1_corr3}
\begin{split}
\lim\limits_{n\rightarrow \infty} \frac{K(\boldsymbol{\lambda}(0)){\rm e}^{Ct}}{(n-\mu)^2}=0.
\end{split} 
\end{equation}
Thus, \eqref{appeq:lemm1_corr1} leads to \eqref{appeq:t_leq_tn_lamn}

\begin{lemma}\label{app:lem1}(Problem 5.3.15 of\cite{karatzas_brownian_1996})
    Suppose $b_i(t,y)$ and $\sigma_{ij}(t,y)$; $1\leq i \leq d$, $1\leq j \leq r$, are progressively measurable functionals from $[0,\infty)\times C[0,\infty)^d$ into $\mathbb{R}$ satisfying
    \begin{equation}\label{appeq:lem1_cond}
    \begin{split}
        \| \bm{b}(t,y)\|^2+\| \bm{\sigma}(t,y)\|^2\leq K\left( 1+\|y\|^2\right);&\\
        \forall 0\leq t< \infty,\ y\in \mathbb{R}^d,&
    \end{split}
    \end{equation}
    where K is a positive constant. If $(\bm{X},W)$, $(\Omega,\mathcal{F},P)$, $\{\mathcal{F}_t\}$ is a weak solution to the SDE 
    \begin{equation}
    \begin{split}
        {\rm d}\bm{X}=\bm{b}(t,\bm{X}){\rm d}t+\bm{\sigma}(t,\bm{X}){\rm d}\bm{W}_t,
    \end{split}
    \end{equation}
    with $\left<\|\bm{X}_0\|^{2m} \right><\infty$ for some $m>1$, then for any finite time $T>0$, we have
    \begin{equation}
    \begin{split}
        \left<\max\limits_{0\leq s\leq t}\|\bm{X}_t\|^{2m} \right>\leq C\left( 1+\left<\|\bm{X}_0\|^{2m} \right>\right){\rm e}^{Ct};0\leq t \leq T,
    \end{split}
    \end{equation}
    where C is a positive constant depend only on $m$, $T$, $K$ and $d$.
\end{lemma}

\section{ Why calculate \texorpdfstring{$<\lambda_{S}(t)^n>$ }instead of \texorpdfstring{$<\exp(\lambda_{S})>$}?}\label{apsec:why}

In this appendix we answer why we calculated $<\lambda_{S}(t)^n>$ instead of $<\exp(\lambda_{S})$.

In the main text, we calculated $<\lambda_{S}(t)^n>$ because it can be determined from solution of a finite number of ODEs, 
whereas calculating $<\exp(\lambda_{S})>$ involves solving an infinite-dimensional system of coupled ODEs.

The details are as follows.  

A natural idea to calculate $\left<\exp(\lambda_{S})\right>$ is to derive its governing ODE by differentiating $\exp(\lambda_{S})$:
\begin{equation}
\begin{split}\label{infi_lam_ode1}
{\rm d}\exp(\lambda_{S})=\exp(\lambda_{S}){\rm d}\lambda_{S}+\frac{1}{2}\exp(\lambda_{S})({\rm d}\lambda_{S})^2.
\end{split}
\end{equation} 

Substituting ${\rm d}\lambda_{S}$ and ${\rm d}\lambda_{S}^2$ in \eqref{infi_lam_ode1} with the SDEs in \eqref{sde_formu1}, some polynomial functions of $\boldsymbol{\lambda}$ show up in the right-hand side of \eqref{infi_lam_ode1}. The new polynomial functions have a general format of  $\boldsymbol{\lambda}^{\bm{l}}\exp(\lambda_{S})$, where $\bm{l}=[l_1,\dots,l_N]$ is a vector of positive integer numbers.

To determine terms like $\boldsymbol{\lambda}^{\bm{l}}\exp(\lambda_{S})$, we kept deriving their ODEs by differentiating them. This process results in terms with higher exponents of $\boldsymbol{\lambda}$, i.e. $\boldsymbol{\lambda}^{\bm{l'}}\exp(\lambda_{S})$, where $\|\bm{l'}\|_1>\|\bm{l}\|_1$: 

\begin{align}\label{infi_lam_fhold}
&{\rm d}(\exp(\lambda_{S})\boldsymbol{\lambda}^{\bm{l}})=\exp(\lambda_{S})\boldsymbol{\lambda}^{\bm{l}}{\rm d}\lambda_{S}+l_1\exp(\lambda_{S})\boldsymbol{\lambda}^{\bm{l}-e_1}{\rm d}\lambda_1+l_2\exp(\lambda_{S})\boldsymbol{\lambda}^{\bm{l}-e_2}{\rm d}\lambda_2+\dots\nonumber\\
&+\frac{1}{2}\exp(\lambda_{S})\boldsymbol{\lambda}^{\bm{l}}({\rm d}\lambda_{S})^2+\frac{l_1(l_1-1)}{2}\exp(\lambda_{S})\boldsymbol{\lambda}^{\bm{l}-2e_1}({\rm d}\lambda_1)^2+\dots\\&+ l_1\exp(\lambda_{S})\boldsymbol{\lambda}^{\bm{l}-e_1}{\rm d}\lambda_{S}{\rm d}\lambda_1+l_2\exp(\lambda_{S})\boldsymbol{\lambda}^{\bm{l}-e_2}{\rm d}\lambda_{S}{\rm d}\lambda_2+l_1l_2\exp(\lambda_{S})\boldsymbol{\lambda}^{\bm{l}-e_2-e_2}{\rm d}\lambda_1{\rm d}\lambda_2+\dots\nonumber,
\end{align}

\noindent where $e_j$ ($j=1,2,\dots,N$) is a $N$ dimensional vector, whose $j$'th component equals one, and all other components equal zero. 

It can be seen from the right-hand side of \eqref{infi_lam_fhold}, that the differentiation process always results 
in nonlinear terms with higher exponents of $\boldsymbol{\lambda}$, whose expectations are not zero, e.g. 
the terms like $\exp(\lambda_{S})\boldsymbol{\lambda}^{\bm{l}}{\rm d}\lambda_{S}^2$. 

As a result, the expectations of \eqref{infi_lam_ode1} and \eqref{infi_lam_fhold} compose a set of non-closed-form equations, and directly solving $<\exp(\lambda_{S})>$ is dogged by infinite-dimensional ODEs.

To avoid solving infinite number of ODEs for $<\exp(\lambda_{S})>$, we approximated it by calculating $<\lambda_{S}(t)^n>$. This is achieved by introducing a new function $H(s,t) = <\exp(s*\lambda_{S})>$, and one may approximate $H(s,t)$ by its Taylor expansions: 

\begin{equation}
\begin{split}\label{pade_taylor}
H(s,t)\approx T_{\tilde{N}}(s,t)=\frac{s^n}{n!} \times\frac{\partial^n}{\partial s^n}H(s,t)|_{s=0}
=\sum_{n=0}^{\tilde{N}}\frac{s^n}{n!} <\lambda_{S}(t)^n>. 
\end{split}
\end{equation}

 \eqref{pade_taylor} means that we only need to determine $<\lambda_{S}(t)^n>$ for every $n$ to determine $H(s,t)$, and $H(1,t)$ is the required FPT distribution. 

For practical reasons, we can calculate $<\lambda_{S}(t)^n>$, with $n$ less than a predefined integer, i.e. $n=0,1,2\dots,\tilde{N}$ ($\tilde{N}$ is the highest order of calculated moments). Since $\exp(s\lambda_{S}(t))$ is a transcendental function, we used a Pad\'e approximant to control
the approximation errors. We calculated the Pad\'e approximant via the extended Euclidean algorithm. 

It turns out that calculating $<\lambda_{S}(t)^n>$ is much easier than calculating $\left<\exp(\lambda_{S})\right>$, because the \textit{linearity} of the SDEs in \eqref{sde_formu1} guarantes that $<\lambda_{S}(t)^n>$ is governed by a set of finite dimensional ODEs. 

We derived the ODEs governing $<\lambda_{S}(t)^n>$ through a similar process as that of $\left<\exp(\lambda_{S})\right>$. We
started differentiating $<\lambda_{S}(t)^n>$ by Ito's rule: 
\begin{equation}
\begin{split}\label{fi_lam_fode1_cop}
{\rm d}(\lambda_{S}^n)&=n\lambda_{S}^{n-1}{\rm d}\lambda_{S}+ \frac{n(n-1)}{2}\lambda_{S}^{n-2}({\rm d}\lambda_{S})^2.\\
\end{split}
\end{equation}

Substituting ${\rm d}\lambda_{S}$ and ${\rm d}\lambda_{S}^2$ in \eqref{fi_lam_fode1_cop} with the SDEs in \eqref{sde_formu1}, the right-hand side of \eqref{fi_lam_fode1_cop} would only be polynomial functions of $\lambda_{S}^{l_0}\boldsymbol{\lambda}^{\bm{l}}$. The linearities of the SDEs guarantee that
${\rm d}\lambda_{S}$ can be represented by polynomial functions of $\bm{\lambda}$ with orders no more than one. Likewise, ${\rm d}\lambda_{S}^2$ can be represented with polynomial functions of $\bm{\lambda}$ with orders no more than 2. As a result, the polynomial order of $\lambda_{S}^{l_0}\boldsymbol{\lambda}^{\bm{l}}$ is not higher than $n$, i.e. $l_0+\|\bm{l}\|_1 \le n$. 

Furthermore, the differentiation of $<\lambda_{S}^{l_0}\boldsymbol{\lambda}^{\bm{l}}>$ does not increase its polynomial order. Because of the It\^o's rule, 
\begin{align}\label{fi_lam_ito}
&{\rm d}(\lambda_{S}^{l_0}\boldsymbol{\lambda}^{\bm{l}})=l_0\lambda_{S}^{l_0-1}\boldsymbol{\lambda}^{\bm{l}}{\rm d}\lambda_{S}+l_1\lambda_{S}^{l_0}\boldsymbol{\lambda}^{\bm{l}-e_1}{\rm d}\lambda_1+l_2\lambda_{S}^{l_0}\boldsymbol{\lambda}^{\bm{l}-e_2}{\rm d}\lambda_2\dots\nonumber\\
&+\frac{l_0(l_0-1)}{2}\lambda_{S}^{l_0-2}\boldsymbol{\lambda}^{\bm{l}}({\rm d}\lambda_{S})^2+\frac{l_1(l_1-1)}{2}\lambda_{S}^{l_0}\boldsymbol{\lambda}^{\bm{l}-2e_1}({\rm d}\lambda_1)^2\\&+ l_0l_1\lambda_{S}^{l_0-1}\boldsymbol{\lambda}^{\bm{l}-e_1}{\rm d}\lambda_{S}{\rm d}\lambda_1+l_0l_2\lambda_{S}^{l_0-1}\boldsymbol{\lambda}^{\bm{l}-e_2}{\rm d}\lambda_{S}{\rm d}\lambda_2+l_1l_2\lambda_{S}^{l_0}\boldsymbol{\lambda}^{\bm{l}-e_2-e_2}{\rm d}\lambda_1{\rm d}\lambda_2\dots,\nonumber
\end{align}
where $e_j$ ($j=1,2,\dots,N$) is a $N$ dimensional vector, whose $j$'th component equals one, and all other components equal zero. Similarly, because of the linearity of \eqref{sde_formu1}, the polynomial orders of $\lambda_{S}$ and $\boldsymbol{\lambda}$ of all terms in \eqref{fi_lam_ito} do not increase.

Therefore, \eqref{fi_lam_fode1_cop} and \eqref{fi_lam_ito} compose a set of closed-form equations, and $<\lambda_{S}(t)^n>$ can be solved by a set of finite-dimensional ODEs.

\end{document}